\documentclass[journal,draftcls,onecolumn,12pt,twoside]{IEEEtranTCOM}

\usepackage{amsmath}
\usepackage{comment}
\usepackage{amssymb}
\usepackage{amsthm}
\usepackage{graphicx}
\usepackage{float}
\usepackage{array}
\usepackage{multirow}
\usepackage{url}
\usepackage{tikz}
\usepackage{pgfplots}
\usepackage{grffile}
\usepackage{arydshln} %for sashes in arrays
\usepackage{comment}
\usepackage{soul} %for strikethrough \st
\usepackage[ruled,vlined]{algorithm2e}

\usetikzlibrary{positioning, calc,arrows,plotmarks}
\usepackage[noadjust]{cite}

\usepackage{hyperref}
\hypersetup{colorlinks=true, unicode=true, linkcolor=[rgb]{0.10,0.05,0.67}, citecolor=[rgb]{0.10,0.05,0.67}, filecolor=[rgb]{0.10,0.05,0.67}, urlcolor=[rgb]{0.10,0.05,0.67}}

\usepackage{lipsum} %footnote without number

\newtheorem{theorem}{Theorem}
\newtheorem{lemma}{Lemma}
\newtheorem{definition}{Definition}
\newtheorem{remark}{Remark}
\newtheorem{example}{Example}

\newtheorem{proposition}[lemma]{Proposition}
\newtheorem{fact}[lemma]{Fact}

\usepackage{environ}
\NewEnviron{killcontents}{}

\begin{document}
\bstctlcite{IEEEexample:BSTcontrol}

%%Eshed - remove this to include proofs
%\let\proof\killcontents
%\let\endproof\endkillcontents

\title{A Unified Spatially Coupled Code Design: Threshold, Cycles, and Locality}

\author{Homa~Esfahanizadeh$^{*}$, Eshed~Ram$^{*}$, Yuval~Cassuto,~\IEEEmembership{Senior~Member,~IEEE}, and~Lara~Dolecek,~\IEEEmembership{Senior~Member,~IEEE}\vspace{-2.4cm}
\date{}
\markboth{IEEE Transactions on Communications}{Submitted paper}
\thanks{$*$ H.~Esfahanizadeh and E.~Ram equally contributed to this paper. H.~Esfahanizadeh is with the Electrical Engineering and Computer Science Department, Massachusetts Institute of Technology (MIT), Cambridge, MA 02139 USA (email: {homaesf}@mit.edu). E.~Ram and Y.~Cassuto are with Andrew and Erna Viterbi Department of Electrical and computer Engineering, Technion, Haifa, Israel (e-mail:
{s6eshedr}@technion.ac.il; {ycassuto}@ee.technion.ac.il). L.~Dolecek is with the Department of Electrical and Computer Engineering, University of California, Los Angeles, Los Angeles, CA 90095 USA (e-mail: {dolecek}@ee.ucla.edu).}
}
\maketitle

\begin{abstract} 

Spatially-Coupled (SC)-LDPC codes are known to have outstanding error-correction performance and low decoding latency. Whereas previous works on LDPC and SC-LDPC codes mostly take either an asymptotic or a finite-length design approach, in this paper we present a unified framework for jointly optimizing the codes’ thresholds and cycle counts to address both regimes. The framework is based on efficient traversal and pruning of the code search space, building on the fact that the performance of a protograph-based SC-LDPC code depends on some characteristics of the code’s partitioning matrix, which by itself is much smaller than the code’s full parity-check matrix. We then propose an algorithm that traverses all nonequivalent partitioning matrices, and outputs a list of codes, each offering an attractive point on the trade-off between asymptotic and finite-length performance. We further extend the framework to designing SC-LDPC codes with sub-block locality, which is a recently introduced feature offering fast access to sub-blocks within the code block. Our simulations show that our framework results in SC-LDPC codes that outperform the state-of-the-art constructions, and that it offers the flexibility to choose low-SNR, high-SNR, or in-between SNR region as the primary design target.

\end{abstract}

\section{Introduction}
Spatially-coupled low-density parity-check (SC-LDPC) codes are a family of graph-based codes that have attracted a lot of attention thanks to their capacity-approaching performance and low-latency decoding \cite{FelstromIT1999,GallagerLDPC1962}. SC-LDPC codes are constructed by coupling  a series of disjoint block LDPC codes into a single coupled chain. We use circulant-based (CB) LDPC codes as the underlying LDPC block codes due to their simple implementation \cite{TannerIT2004}. 
SC-LDPC codes are known to have many desirable properties, such as threshold saturation \cite{KudekarRichUrb13} and linear-growth of the size of minimal trapping sets in typical codes from an ensemble \cite{Mitchell11}. 
These properties, respectively, imply good bit-error rate (BER) performance in the \emph{waterfall} and \emph{error floor} regions, using the belief-propagation (BP) decoder. 
Recently, \cite{RamCass18b} introduced SC-LDPC codes with a \emph{sub-block locality} feature, where in addition to the usual decoding, the codes can be decoded \emph{locally} in small sub-blocks for fast read access. 

From the asymptotic perspective, density evolution (DE) techniques have been used to study the decoding threshold of SC-LDPC codes, e.g., \cite{LentmaierIT2010,KudekarRichUrb13}, among others. From the finite-length perspective, methodologies for the evaluation and optimization of the number of problematic combinatorial objects are studied, e.g., \cite{EsfahanizadehTCOM2018,AhmedCycles8,MitchellISIT2017,BeemerISITA2016}.
The asymptotic properties (e.g., decoding threshold) of LDPC codes are the dominant performance determinants in the low-SNR region, and the finite-length properties (e.g., number of short cycles) are the dominant ones in the error-floor (high-SNR) region \cite{6928457,article9273663}. 
This is because in the low-SNR region, the performance is typically dominated by the properties of the code’s tree ensemble, while in the high-SNR region the performance depends on the incidence of problematic combinatorial objects in the code’s graph \cite{BaneBook, stoppingset,TomasoniAS,UrbankeFLAnalysis,KarimiBanihashemi2011,RichardsonErrorFloor}.

%In fact, the asymptotic methods are based on assumptions, e.g., being cycle-free, that do not hold in the finite-length setting. 
%Thus, good asymptotic properties obtained via DE techniques do not necessarily guarantee a good performance in practice, particularly in the high SNR region \cite{BaneBook}. 
%Instead, finite-length measures such as the number of problematic combinatorial objects in the graphs representing the codes \cite{stoppingset,TomasoniAS,UrbankeFLAnalysis,KarimiBanihashemi2011} predict the performance in the high SNR values as these combinatorial objects are responsible for most of the errors that occur in this region \cite{RichardsonErrorFloor}.

{%There are different approaches to designing SC-LDPC codes, many involving a certain use of non-rigorous tools such as fully-heuristics and simulations. 
In this paper, we pursue a principled comprehensive approach that combines various design metrics, e.g., targeted SNR region and decoding latency. In fact, any future efficient code design can be incorporated into our proposed framework. Here, we evaluate each candidate code in terms of its exact threshold and cycle-counts, and extract a small set of attractive codes from a large pool of initial candidates. This approach works best when we can examine all possible codes of a given set of design parameters. However, the resulting space of candidates may blow up quickly, thus requiring efficient traversal and pruning methods. One particularly effective such method is to avoid multiple candidates that are equivalent in terms of performance. The elimination of this multiplicity can reduce the candidate list -- and in turn the complexity of code design -- by orders of magnitude.} 

{To enable the aforementioned comprehensive design approach, we formalize in this paper the notion of \textit{performance-equivalent} SC-LDPC codes. It is well-known that two linear codes are equivalent in terms of most performance figures if one's parity-check matrix can be obtained from the other's by a sequence of row and column permutations. In Section~\ref{Sec:CC_opt} we prove that for SC-LDPC codes, the same property holds for the code's {\em partitioning matrix}, which is much smaller than the full (coupled+lifted) code matrix. This motivates our exact derivation in Section~\ref{Sec:equivalence} of the number of nonequivalent binary matrices with up to 3 rows, which captures all regular SC-LDPC constructions with unit memory and up to $3$ check nodes in the uncoupled protograph (regular unit-memory codes have binary partitioning matrices). This exact count can be easily translated to an efficient traversal of all nonequivalent matrices, thus enabling an efficient code design. Indeed, we detail in Section~\ref{sub:algo:trade-off} a joint threshold+cycle code-design algorithm, which outputs a final list of candidates with the property that each candidate has: 1) the best threshold among all codes with equal or better cycle-counts, and 2) the best cycle-count among all codes with equal or better thresholds.}

{We extend the joint threshold+cycle design approach to SC-LDPC codes with sub-block locality in Section~\ref{Sec:localglobal_des}. Adding the local-code structure to the partitioning matrix may significantly increase the code design search space and complexity. We exploit the sub-block structure to further reduce the code-design complexity: we propose to replace the threshold calculation by an efficiently computable proxy threshold, and prove analytical results on the threshold and cycle-counts of two natural irregular structures of the local code. Finally, we show simulation results of codes from our proposed design algorithms, for both standard codes (no locality) and codes with sub-block locality. Our codes are shown to outperform prior constructions based on cutting-vector and optimal-overlap partitioning.} {The code and data used in the paper are available for public access at \href{https://github.com/hesfahanizadeh/Unified_SC_LDPCL/}{https://github.com/hesfahanizadeh/Unified\_SC\_LDPCL/}}.

%\st{The rest of this paper is organized as follows. In Section~{II}, we state the necessary preliminaries. In Section~{III}, we present our new methodology for reducing the search space of LDPC codes. In Section~{IV}, we introduce a new algorithm for designing a minimal set of SC-LDPC codes that offer a trade-off between their waterfall and error floor performances. In Section~{V}, we extend the framework for designing SC codes with sub-block locality and also show that the locality constraints reduce the search space even more. In Section~{VI}, simulation results are discussed.  Finally, Section~{VII} is dedicated to conclusions and future work.}

\section{Preliminaries}
Throughout this paper, matrices, vectors, and scalars are represented by uppercase bold letters (e.g., $\mathbf{A}$), lowercase italic letters with an overline (e.g., $\overline{a}$), and lowercase italic letters (e.g., $a$), respectively. Sets and functions are represented by calligraphic italic letters (e.g., $\mathcal{A}$) and uppercase italic letters (e.g., $A(\cdot)$), respectively. The matrix transpose operation, the cardinality of a set, and the factorial function are denoted by $(\cdot)^T$, $|\cdot|$, and $(\cdot)!$, respectively. 
The notation $\mathbf{A}=[a_{i,j}]$ refers to a matrix $\mathbf{A}$ where $a_{i,j}$ is the element in row $i$ and column $j$. We denote the all-one and all-zero matrices with size $m\times n$ as $\mathbf{1}_{m\times n}$ and $\mathbf{0}_{m\times n}$, respectively.

\subsection{{Construction of SC-LDPC Codes}}\label{Sub:SC-LDPC}

An LDPC protograph is a small bipartite graph represented by a $\gamma\times\kappa$ bi-adjacency proto-matrix $\mathbf{B}=[b_{i,j}]$ (where $\gamma$ and $\kappa$ are positive integers and $\gamma<\kappa$), i.e., there is an edge between check node (CN) $i$ and variable node (VN) $j$ if and only if $b_{i,j}=1$. In general, $b_{i,j}>1$ (parallel edges) are allowed in the protograph. In this work, without loss of generality\footnote{This is because parallel edges can be avoided by duplicating protograph nodes.}, we focus on $b_{i,j}\in \{0,1\}$.
A sparse parity-check matrix $\mathbf{H}$ (or its corresponding representation as a Tanner graph) is generated from $\mathbf{B}$ by a \textit{lifting} operation with a positive integer $z$ that is called the circulant size. 
The rows (resp. columns) of $\mathbf{H}$ corresponding to row $i\in\{1,\ldots,\gamma\}$ (resp. column $j\in\{1,\ldots,\kappa\}$) of $\mathbf{B}$, are called row group $i$ (resp. column group $j$). 

In this paper, we use \emph{circulant-based} (CB) lifting \cite{TannerIT2004}, where the circulants, each with size $z\times z$, are either zero or an identity matrix shifted by a certain number of units to the left, described by the \textit{circulant power}.
%In other words, a non-zero circulant is a single-shift identity matrix raised by a power. 
The powers of the circulants are represented by the \textit{power matrix} $\mathbf{C}=[c_{i,j}]$ of size $\gamma\times\kappa$, such that the non-zero elements in row group $i$ and column group $j$ in $\mathbf{H}$ form a single-shift identity matrix raised to the power $c_{i,j}$. In our simulations, the power matrix $\mathbf{C}=[c_{i,j}]$ is defined such that $c_{i,j}=\alpha{\cdot}i{\cdot}j$, for a constant positive integer $\alpha$. 
This choice ensures that no length-$4$ cycle (cycle-$4$ for short) exists when the circulant size is a prime number \cite{LaraAB}. Thus, this paper focuses on length-$6$ cycles (cycles-$6$ for short) as the most problematic cycles.

Let proto-matrix $\mathbf{B}$ be the parity-check matrix of a protograph block code. The matrix of an SC-LDPC protograph \cite{FelstromIT1999} with memory $m$ and coupling length $l$ is constructed from $\mathbf{B}$ by partitioning it into $m+1$ matrices $\mathbf{B}_0$,\dots,$\mathbf{B}_m$ such that $\mathbf{B}=\sum_{k=0}^{m}\mathbf{B}_k$, and stacking $l$ \textit{replicas} of $[\mathbf{B}_0;\mathbf{B}_1;\dots;\mathbf{B}_m]$ (where `;' represents vertical concatenation here) on the diagonal of the coupled proto-matrix $\mathbf{B}_\text{SC}$. For proto-matrix $\mathbf{B}$ of size $\gamma\times\kappa$, the resulting coupled proto-matrix $\mathbf{B}_\text{SC}$ has size $(l+m)\gamma\times l\kappa$.
We represent this partitioning by a matrix $\mathbf{P}=\left[p_{i,j}\right]$, called \textit{partitioning matrix}, where $p_{i,j}\in\{0,1,\dots,m,{\star}\}$. If $p_{i,j}={\star}$, then there is a zero in row $i$ and column $j$ of $\mathbf{B}$. Otherwise, the non-zero element is assigned to $\mathbf{B}_{p_{i,j}}$. This description captures both regular and irregular SC constructions.
In this work, we focus on SC codes with $m=1$, thus the partitioning operation determines which (non-zero) elements are assigned to $\mathbf{B}_0$ and which ones are assigned to $\mathbf{B}_1$ (when referring to lifted graphs, we use $\mathbf{H}_0$ and $\mathbf{H}_1$). 

\subsection{{Extension to Codes with Sub-Block Locality (SC-LDPCL)}}\label{Sub:SC-LDPCL}
One drawback of SC codes is their typically large block size, which increases the complexity and latency of decoding. To mitigate this obstacle, SC-LDPC codes with \emph{sub-block locality} were developed and studied in \cite{EshedTIT}. In these codes, the block is partitioned into smaller sub-blocks, and the code is designed for both sub-block and full-block decoding. The former, called \emph{local decoding}, allows low-latency access, and the latter, called \emph{global decoding}, provides the usual high-reliability of SC-LDPC codes. There is flexibility in how to partition the code block to sub-blocks, and in this paper we define each sub-block to comprise all the VNs in one replica of the coupled SC code (hence the parameter $l$ is also the number of sub-blocks). In local decoding, only CNs that are connected solely to VNs within the sub-block are used, and we call them \emph{local CNs} (LCNs). All other CNs are called \emph{coupling CNs} (CCNs) \cite{EshedTIT}. In another view, rows in $\mathbf{P}$ that have both $0$ and $1$ entries result in CCNs in the coupled matrix; we mark the number of such rows as $\gamma_c$. The LCNs are specified by $\gamma_l\triangleq \gamma-\gamma_c$ rows in $\mathbf{P}$ that do not mix $0$ and $1$ entries, and are designed to induce a non-zero asymptotic local-decoding threshold \cite{EshedTIT}. Without loss of generality, the rows of $\mathbf{P}$ are ordered such that the first $\gamma_c$ rows correspond to CCNs.\vspace{-0.3cm}

\subsection{Asymptotic Analysis of Protographs: The EXIT Method}\label{Sub:EXIT}
The EXtrinsic Information Transfer (EXIT) method \cite{TenBrink04} is a useful tool for analyzing and designing LDPC codes in the asymptotic regime over the AWGN channel with the channel parameter $\sigma$. Let $ J\colon [0,\infty)\to [0,1) $ be a function that represents the mutual information between the channel input and a corresponding message passing in the Tanner graph.
For a VN of degree $d_v$ in the protograph, with incoming EXIT values $\{J_{i}\}_{i=1}^{d_v-1}$, the VN$\rightarrow$CN EXIT value is
\begin{align}\label{Eq:J VN}
J_{out}^{(V)}\!\left (s_{ch},J_{1},\ldots,J_{d_v\!-\!1}\right )\!=\!J\!\!\left (\!\sqrt{\sum_{i=1}^{d_v-1}\!\!\left (J^{-1}(J_{i})\right )^2\!+\!s_{ch}^2}\right )\!,
\end{align}
where $s_{ch}^2=4/\sigma^2$.
For a CN of degree $d_c$ in the protograph with incoming EXIT values $\{J_{j}\}_{j=1}^{d_c-1}$, the CN$\rightarrow$VN EXIT value is approximated by
\begin{align}\label{Eq:J CN}
\begin{split}
J_{out}^{(C)}\!(J_{1},\ldots,J_{d_c-1})\! =\! 1\!-\!J_{out}^{(V)}\!\left (0,1\!-\!J_{1},\ldots,1\!-\!J_{d_c\!-\!1}\right ).
\end{split}
\end{align}
The functions $J_{out}^{(V)}$ and $J_{out}^{(C)}$ are monotonically non decreasing with respect to all their arguments.
In simulations, we use approximations of $J(\cdot)$ and $J^{-1}(\cdot)$  \cite{TenBrink04}. 
By alternately applying \eqref{Eq:J VN} and \eqref{Eq:J CN} for every edge in a protograph and by varying $\sigma$, a threshold value $ \sigma^*$ can be found such that all EXIT values on VNs approach $ 1 $ as the number of iterations increases if and only if $ \sigma<\sigma^* $ \cite{Liva}. We mark the threshold of a protograph $\mathbf{B}$ by $ \sigma^*(\mathbf{B}) $.

\subsection{Short-Cycle Optimization and Overlap Parameters}
\label{Sub:cycles_opt}
Short cycles have a negative impact on the performance of block-LDPC and SC-LDPC codes under BP decoding: 1) they affect the independence of the messages exchanged on the graph, 2) they enforce upper-bounds on the minimum distance, and 3) they form combinatorial objects in the Tanner graphs that fail the iterative decoder in different known ways \cite{SmarandacheIT2012,EsfahanizadehTCOM2018}.
\begin{definition}\label{def_ov_par}
Consider a binary matrix $\mathbf{B}$. A degree-$d$ overlap parameter $t_{\{i_1,\dots,i_d\}}$ is the number of columns in which all rows of $\mathbf{B}$ indexed by ${\{i_1,\dots,i_d\}}$ have $1$s.
\end{definition}

The overlap parameters contain all the information we need to find the number of cycles in the graph represented by the matrix. We are particularly interested in cycles-$6$ (i.e., cycles with 6 nodes), as they are the shortest cycles for practical LDPC codes (most practical high-rate LDPC codes, in particular the codes in this paper, are designed to have girth at least $6$). Consider a binary matrix $\mathbf{B}$ with $\gamma$ rows and $\kappa$ columns. The number of cycles-$6$ in the graph of matrix $\mathbf{B}$ can be expressed in terms of the overlap parameters of matrix $\mathbf{B}$ as follows: 
\begin{equation}\label{eq:cyclesA}
    F(\mathbf{B})=\sum_{\{i_1,i_2,i_3\}\subseteq \{1,\dots,\gamma\}}\hspace{-0.8cm}A(t_{\{i_1,i_2,i_3\}},t_{\{i_1,i_2\}},t_{\{i_1,i_3\}},t_{\{i_2,i_3\}})\,,
\end{equation}
where $A$ is given by (see \cite{EsfahanizadehTCOM2018})
\begin{equation}\label{equ_A}\vspace{-0.0cm}
\begin{split}
A&(t_{\{i_1,i_2,i_3\}},t_{\{i_1,i_2\}},t_{\{i_1,i_3\}},t_{\{i_2,i_3\}})=\left(t_{\{i_1,i_2,i_3\}}[t_{\{i_1,i_2,i_3\}}-1]^+[t_{\{i_2,i_3\}}-2]^+\right)\\
+&\left(t_{\{i_1,i_2,i_3\}}(t_{\{i_1,i_3\}}-t_{\{i_1,i_2,i_3\}})[t_{\{i_2,i_3\}}-1]^+\right)+\left((t_{\{i_1,i_2\}}-t_{\{i_1,i_2,i_3\}})t_{\{i_1,i_2,i_3\}}[t_{\{i_2,i_3\}}-1]^+\right)\\
+&\left((t_{\{i_1,i_2\}}-t_{\{i_1,i_2,i_3\}})(t_{\{i_1,i_3\}}-t_{\{i_1,i_2,i_3\}})t_{\{i_2,i_3\}}\right),\vspace{-0.2cm}
\end{split}
\end{equation}
and, $[\alpha]^+\triangleq\max\{\alpha,0\}$. The optimization problem for identifying the optimal overlap parameters, and consequently the optimal partitioning, for designing SC-LDPC protographs with minimum number of cycles-$6$ is presented in \cite{EsfahanizadehTCOM2018}. The approach is called the optimal overlap (OO) partitioning, and is one of the baselines in our experimental results. %In this paper, we introduce a novel discrete optimization problem that has a smaller search space than overlap parameters, fully encompasses all necessary design choices, and has a meaningful connection with the overlap parameters.

\section{Reducing Search Space: Equivalent Binary Matrices}\label{Sec:equivalence}
%\vspace{-0.2cm}

In this section, we explore the space of
all possible binary matrices
of a given size $\gamma\times\kappa$. {In the next section, these binary matrices will correspond to partitioning matrices defining the SC-LDPC codes, but in the meantime, it will be instructive to think about these matrices as parity-check matrices of some protograph-based code (called proto-matrices).} We introduce a combinatorial representation that allows to significantly reduce the search space size by capturing the equivalence among the codes and only keeping one candidate from each class of equivalent codes. Equivalent codes are codes whose proto-matrices can be obtained from one another by a sequence of row and column permutations. This equivalence definition is motivated by the fact that permutation of rows and columns in proto-matrices
affect neither the asymptotic threshold nor the number of cycles in the protograph. 
In Section~~\ref{Sec:CC_opt}, we theoretically prove that equivalence under this definition for partitioning matrices implies performance-equivalent SC-LDPC codes. %as permutations of rows and columns of the partitioning matrix do not affect these metrics in the SC protograph. 
Using a new technique for representing the binary matrices, 
we only consider one code from each equivalence class, thereby significantly reducing the search complexity.\vspace{-0.2cm}

%Now, we propose our reduced search space for designing binary matrices. We first define the notation for row and column permutations of a matrix.

\begin{definition}[Column and Row Permutation]\label{def:perm_defs}
A column permutation of matrix $\mathbf{A}$ with $\kappa$ columns is denoted by a vector $\overline{\pi}_c=[\rho_1,\dots,\rho_{\kappa}]$, that is a permutation of the elements in the vector $[1,\dots,\kappa]$. When $\mathbf{A}\xrightarrow[]{\overline{\pi}_c}\mathbf{A}'$, $\mathbf{A}'$ is obtained from $\mathbf{A}$ such that the $j$-th column of $\mathbf{A}$ is the $\rho_j$-th column of $\mathbf{A}'$. Similarly, a row permutation of matrix $\mathbf{A}$ with $\gamma$ rows is denoted by a vector $\overline{\pi}_r=[\nu_1,\dots,\nu_{\gamma}]$, that is a permutation of the elements in the vector $[1,\dots,\gamma]$. When $\mathbf{A}\xrightarrow[]{\overline{\pi}_r}\mathbf{A}'$, $\mathbf{A}'$ is obtained from $\mathbf{A}$ such that the $i$-th row of $\mathbf{A}$ is the $\nu_i$-th row of $\mathbf{A}'$.\vspace{-0.2cm}
\end{definition}

\begin{definition}[Equivalent Matrices]\label{Def: equivalenc}
Two binary matrices $\mathbf{A}$ and $\mathbf{A}'$ are column-wise equivalent (resp., row-wise equivalent) if they can be obtained by column (resp., row) permutations of each other. Two binary matrices $\mathbf{A}$ and $\mathbf{A}'$ are equivalent if they can be derived from each other by a sequence of row and column permutations.\vspace{-0.2cm}
\end{definition}

\begin{remark}
We use the notion of equivalence in this paper to highlight that neither decoding threshold nor the number of combinatorial objects, e.g., cycles, absorbing sets \cite{MitchellISIT2014}, trapping sets \cite{6847681}, etc., change with a sequence of row and column permutations of a binary matrix. %\footnote{In fact, row and column permutations correspond to renaming the nodes in the Tanner graph.}
\end{remark}

We present an efficient combinatorial approach for identifying the nonequivalent binary matrices given their size. Consider a binary matrix $\mathbf{B}$ with size $\gamma\times\kappa$. There are $2^\gamma$ distinct choices for each column of $\mathbf{B}$, i.e., $[0,0,\dots,0]^T$, $[0,0,\dots,1]^T$, \dots, $[1,1,\dots,1]^T$. The set of all binary matrices with size $ \gamma\times\kappa $ is of cardinality $ 2^{\gamma\kappa} $. In what follows, we show that the search space is effectively much smaller due to the equivalence among matrices. 
Our goal is to identify the set of nonequivalent binary matrices with $\gamma$ rows and $\kappa$ columns, denoted by $\mathcal{K}_{\kappa,\gamma}$, and to find a closed-form expression for its cardinality. This reduction, as we numerically verify, combined with an algorithm for iterating over nonequivalent binary matrices, allows a significantly more efficient optimization of LDPC protographs in terms of short cycles and thresholds.

\begin{definition}[Column Type]
The type of a column of a binary matrix is defined as the decimal representation of the binary vector with the top element being the most significant bit.
\end{definition}

\begin{definition}[Column Distribution]\label{col_dist}
We associate with matrix $\mathbf{B}\in\{0,1\}^{\gamma\times\kappa}$ a vector 
$\overline{n}(\mathbf{B})= \left [n_0,n_1,\ldots,n_{2^\gamma-1}\right ]$ such that for every $ i\in \{0,1,\ldots,2^{\gamma}-1\} $, $ n_i $ is the number of \emph{columns} in $\mathbf{B}$ with type $i$. 
We call $ \overline{n}(\mathbf{B}) $
the \emph{column distribution} of $\mathbf{B}$, where the term stems from the fact that for every matrix $\mathbf{B}\in\{0,1\}^{\gamma\times\kappa}$, the entries of $ \overline{n}(\mathbf{B}) $ sum up to $ \kappa $, i.e., $ \sum_{i=0}^{2^\gamma-1}n_i = \kappa $. 
\end{definition}

\begin{example}\label{Ex:col dist}
    Consider matrix 
%$\mathbf{B} = 
%(
%   0 \, 1 \, 0 \, 1 \, 1;
%    0 \, 1 \, 1 \, 1 \, 1;
%    1 \, 1 \, 0 \, 1 \, 0
%    )$
$ \mathbf{B} = 
    \left[\begin{array}{ccccc}
    0 & 1 & 0 & 1 & 1\\
    0 & 1 & 1 & 1 & 1\\
    1 & 1 & 0 & 1 & 0\end{array}\right].$
Then, $ \overline{n}(\mathbf{B}) = [0,1,1,0,0,0,1,2]$.
\end{example}

Since column permutations do not change the column distribution of a matrix, we identify the number of column-wise nonequivalent matrices by counting the number of distinct column distributions. 
We note that a family of column-wise nonequivalent matrices can include row-wise equivalent matrices. At the same time, by excluding column-wise equivalent multiplicities, some row-wise equivalent multiplicities will also be excluded. For example, consider the $2\times 2$ binary matrices: for matrix $[1 \;0 ; 1\; 1]$, no column permutation will lead to the row-permuted version $[1\; 1 ; 1\; 0]$; however, for matrix $[1\; 0; 0\; 1]$, swapping the columns will yield $[0\; 1; 1\; 0]$, i.e., swapping rows.
The relation between families of column-wise nonequivalent matrices and row-wise nonequivalent matrices is not trivial, and how to derive the family of nonequivalent matrices is one contribution of this paper.

\subsection{Column-Wise Nonequivalent Binary Matrices} 
In this part, we explore the family of column-wise nonequivalent binary matrices, their connection to the stars-and-bars problem in combinatorics \cite{10.5555/2124415}, and their connection to the overlap parameters in \cite{EsfahanizadehTCOM2018}. 
We also derive a closed-form expression for the number of column-wise nonequivalent matrices and describe how to simply iterate over them to study their properties, e.g., threshold, cycle-counts, etc.

\begin{lemma}\label{lemma:StarsnBars}
	Let $\gamma$ and $\kappa $ be two positive integers, and let $ \mathcal S_{\kappa,\gamma} $ be the set of all distinct column distributions for a $ \gamma\times\kappa$ binary matrix, i.e.,
	\[
	\mathcal S_{\kappa,\gamma} = \left \{ \left [n_0,n_1,\ldots,n_{2^{\gamma}-1}\right]\colon n_i\geq0,\; \sum_{i=0}^{2^\gamma-1}n_i = \kappa\right  \} .
	\]
	Then,
	\begin{align}\label{Eq:StarsnBars}
	\left|\mathcal S_{\kappa,\gamma} \right | = \binom{\kappa+2^\gamma-1}{\kappa} =\binom{\kappa+2^\gamma-1}{2^\gamma-1}.  
	\end{align}
\end{lemma}
\begin{proof}
The proof follows by applying the elementary stars-and-bars method \cite{10.5555/2124415}, where each \textit{bin} represents a column type and the \textit{elements} are the column indices.    
\end{proof}

We highlight that $|\mathcal{S}_{\kappa,\gamma}|$ is the number of column-wise nonequivalent binary matrices of size $\gamma\times\kappa$. 
A recursive algorithm that iterates over all column-wise nonequivalent binary matrices is given in the Appendix as Algorithm~\ref{star_bar_sol}.

\begin{example}\label{Ex:starsnbars}
	Let $ \kappa=11$ and $\gamma=3 $. Then, there are $ \binom{18}{7} = 31,824 $ column-wise nonequivalent binary matrices with $ \gamma $ rows and $ \kappa  $ columns, which is $ 2.7{\cdot}10^5$ times smaller then the entire space $  \{0,1\}^{\gamma \kappa }$. 
\end{example}

In the next subsections, we further reduce the search space by taking into account the row permutations. In Lemma~\ref{lemma:StarsnBars}, we identified the set of distinct column distributions. It is somewhat challenging to calculate how many of these distributions result in equivalent matrices and thus can still be obtained from each other by a sequence of row and column permutations. 
In what follows, we complete this derivation, and to keep the analysis in this paper tractable, we derive the closed-form expressions only for $ \gamma\in\{2,3\} $ ($ \gamma=1 $ case is trivial).

\subsection{Nonequivalent Binary Matrices With $\gamma=2$}
\label{Sub:gamma2}
In the following lemma, we first state necessary and sufficient conditions for two column distributions to correspond to a pair of equivalent matrices. Then, in a subsequent theorem, we show how to use this lemma to reduce the search space of column distributions such that it consists only of those corresponding to nonequivalent matrices, and we identify the cardinality of this reduced search space. 
\begin{lemma}\label{lemma:gamma2}
	Two binary matrices with  $\gamma=2 $ rows and with column distributions $ [n_0,n_1,n_2,n_3] $ and $ [m_0,m_1,m_2,m_3] $ are row-wise equivalent if and only if $n_0=m_0$, $n_3=m_3$, and either $(n_1=m_1 ,n_2=m_2)$ or $(n_1=m_2 ,n_2=m_1)$.
\end{lemma}
\begin{proof}
For $\gamma=2 $, a row permutation is either the identity permutation $\overline{\pi}_r=[1,2]$ or a row swap $\overline{\pi}_r=[2 ,1]$. In the latter, columns with type $0$, i.e., $[0,0]^T$, and type $3$, i.e., $[1,1]^T$, are invariant to the permutation. However, columns with type $1$, i.e., $[0,1]^T$, map to columns with type $2$, i.e., $[1,0]^T$, and vice versa. This concludes the proof.
\end{proof}

In view of Lemma~\ref{lemma:gamma2}, if a column distribution $ [n_0,n_1,n_2,n_3]\in  \mathcal S_{\kappa,2} $ has $ n_1\neq n_2 $, then there exists a different column distribution $ [n_0,n_2,n_1,n_3]\in  \mathcal S_{\kappa,2} $ (i.e., $ n_2 $ and $ n_1 $ are swapped), such that they represent equivalent matrices. This fact readily leads to the following:

\begin{theorem}\label{theorem:gamma2}
	Let $\gamma=2$ and $\kappa $ be a positive integer, and let $ \mathcal K_{\kappa,2} $ be the set of column distributions for all $ 2\times\kappa$ nonequivalent binary matrices. Then,
\begin{subequations}
	\begin{align}\label{Eq:KExplicit1}
	\mathcal K_{\kappa,2} = \left \{ \left[n_0,n_1,n_2,n_3\right]\colon n_i\geq0,\; \sum_{i=0}^{3}n_i = \kappa, \; n_1\leq n_2\right  \},
	\end{align}
	and
	\begin{align}\label{Eq:KExplicit2}
	\left|\mathcal K_{\kappa,2} \right | = \frac12\left (\sum_{i=0}^{\lfloor\kappa/2\rfloor} (\kappa-2i +1) + \binom{\kappa+3}{3}\right ).  
	\end{align}
\end{subequations}
\end{theorem}

\subsection{Nonequivalent Binary Matrices with $\gamma=3$}

Similar to the discussion in Section~\ref{Sub:gamma2}, we first identify necessary and sufficient conditions for two column distributions to correspond to a pair of equivalent matrices. Then, we use the results to enumerate the reduced search space of column distributions.

\begin{lemma}\label{lemma:gamma3}
	Two binary matrices with $\gamma=3$ and with column distributions $[n_0,n_1,\ldots,n_7]$ and $[m_0,m_1,\ldots,m_7]$ are row-wise equivalent if and only if $n_0=m_0$, $n_7=m_7$, and
	\begin{equation*}
	\begin{split}
	&(n_1=m_1,\; n_2=m_2,\; n_4=m_4,\; n_6=m_6,\;n_5=m_5,\;n_3=m_3),\quad\text{or}\\
	&(n_1=m_1,\; n_2=m_4,\; n_4=m_2,\; n_6=m_6,\;n_5=m_3,\;n_3=m_5),\quad\text{or}\\
	&(n_1=m_2,\; n_2=m_1,\; n_4=m_4,\; n_6=m_5,\;n_5=m_6,\;n_3=m_3),\quad\text{or}\\
	&(n_1=m_2,\; n_2=m_4,\; n_4=m_1,\; n_6=m_5,\;n_5=m_3,\;n_3=m_6),\quad\text{or}\\
	&(n_1=m_4,\; n_2=m_1,\; n_4=m_2,\; n_6=m_3,\;n_5=m_6,\;n_3=m_5),\quad\text{or}\\
	&(n_1=m_4,\; n_2=m_2,\; n_4=m_1,\; n_6=m_3,\;n_5=m_5,\;n_3=m_6).\hspace{0.85cm}
	\end{split} 
	\end{equation*}
\end{lemma}

\begin{proof}
For $\gamma=3$, there are eight different column types and $3!=6$ possible row permutations. The proof follows by tracking the changes in column types when applying each of the possible row permutations (from top to bottom): 1) $\overline{\pi}_r=[1,2,3]$ (identity permutation), 2) $\overline{\pi}_r=[2,1,3]$ 3) $\overline{\pi}_r=[1,3,2]$ 4) $\overline{\pi}_r=[2,3,1]$ 5) $\overline{\pi}_r=[3,1,2]$ 6) $\overline{\pi}_r=[3,2,1]$. 
\end{proof}

\begin{theorem}\label{theorem:gamma3}
Let $\gamma=3$ and $\kappa $ be a positive integer, and let $ \mathcal K_{\kappa,3} $ be the set of column distributions for all $ 3\times\kappa$ nonequivalent binary matrices. Then,
	\begin{equation}
	\begin{split}
	    \mathcal K_{\kappa,3}&=\bigg\{ [n_0,n_1,\ldots,n_7] \colon n_i\geq0,\; \sum_{i=0}^{7}n_i=\kappa,\\
	    &(n_1<n_2<n_4)\;\quad\textnormal{or}\quad (n_1=n_2<n_4\quad\textnormal{and}\quad n_6\leq n_5)\quad\textnormal{or}\\
	   &(n_1<n_2=n_4\quad\textnormal{and}\quad n_5\leq n_3)\quad\textnormal{or}\quad (n_1=n_2=n_4\quad\textnormal{and}\quad n_6\leq n_5\leq n_3)\bigg\} ,
	\end{split}
	\end{equation}
	and $\left|\mathcal K_{\kappa,3} \right | = a_\kappa + b_\kappa +c_\kappa$, where
	\begin{align}\label{Eq:KExplicit3}
	\begin{split}
	&a_\kappa = \sum_{i,j\in \mathbb N\colon\atop{ 3(i+j)\leq \kappa}}  (\kappa-3(i+j)+1),\quad b_\kappa = \sum_{i,j\in \mathbb N\colon\atop{  2(i+j)\leq \kappa}} \binom{\kappa-2(i+j)+3}{3} - a_\kappa,\\
	&c_\kappa = \frac16\left (\binom{\kappa+7}{7}-3b_\kappa - a_\kappa \right ).
	\end{split}
	\end{align}
\end{theorem}

\begin{example}\label{Ex:gamma3_2}
	Let $ \kappa=11$ and $\gamma=3 $. Then, there are $ |\mathcal K_{11,3}|= 60+1{,}452+4{,}568=6{,}080 $ nonequivalent binary matrices with $ \gamma $ rows and $ \kappa  $ columns, which is $ 1.41{\cdot} 10^6
 $ times smaller then the cardinality of the entire space  $  \{0,1\}^{\gamma \kappa }$. 
\end{example}

\begin{example}\label{example:searchspace}
\begin{figure}
    \label{fig:search_space_size_final}
    \centering
    %\begin{tabular}{cc}
    \begin{tikzpicture}

\begin{axis}[%
width=2.5in,
height=2.5in,
at={(0,0)},
xmin=5,
xmax=37,
xlabel={$\kappa$},
ymode=log,
ymin=10,
ymax=1e4,
yminorticks=true,
ylabel={Search Space Cardinality},
xmajorgrids,
ymajorgrids,
yminorgrids,
legend columns=-1, 
legend entries={ Exhaustive, OO, Non-isomprphic},
legend to name= leg,
legend style={
    at={(-0.5,-0.3)},anchor=north,
    legend cell align=left
}
]
\addplot [color=red,line width=1pt]
  table[row sep=crcr]{%
5	1024\\
6	4096\\
7	16384\\
8	65536\\
9	262144\\
10	1048576\\
11	4194304\\
12	16777216\\
13	67108864\\
14	268435456\\
15	1073741824\\
16	4294967296\\
17	17179869184\\
18	68719476736\\
19	274877906944\\
};

%%\addlegendentry{Cycles Driven}

\addplot [color=blue,line width=1pt]
  table[row sep=crcr]{%
5	56\\
6	84\\
7	120\\
8	165\\
9	220\\
10	286\\
11	364\\
12	455\\
13	560\\
14	680\\
15	816\\
16	969\\
17	1140\\
18	1330\\
19	1540\\
20  1771\\
21  2024\\
22  2300\\
23  2600\\
24  2925\\
25  3276\\
26  3654\\
27  4060\\
28  4495\\
29  4960\\
30  5456\\
31  5984\\
32  6545\\
33  7140\\
34  7770\\
35  8436\\
36  9139\\
37  9880\\
};
%\addlegendentry{Threshold Driven}

\addplot [color=teal,line width=1pt]
  table[row sep=crcr]{%
5   22\\
6   34\\
7   50\\
8   70\\
9   95\\
10  125\\
11  161\\
12  203\\
13  252\\
14  308\\
15  372\\
16  444\\
17  525\\
18  615\\
19  715\\
20  825\\
21  946\\
22  1078\\
23  1222\\
24  1378\\
25  1547\\
26  1729\\
27  1925\\
28  2135\\
29  2360\\
30  2600\\
31  2856\\
32  3128\\
33  3417\\
34  3723\\
35  4047\\
36  4389\\
37  4750\\
};
\end{axis}

\begin{axis}[%
width=2.5in,
height=2.5in,
at={(2.5in,0)},
xmin=5,
xmax=37,
xlabel={$\kappa$},
ymode=log,
ymin=10,
ymax=1e8,
yminorticks=true,
xmajorgrids,
ymajorgrids,
yminorgrids
]
\addplot [color=red,line width=1pt]
  table[row sep=crcr]{%
5	32768\\
6	262144\\
7	2097152\\
8	16777216\\
9	134217728\\
10	1073741824\\
11	8589934592\\
12	68719476736\\
13	549755813888\\
};

\addplot [color=blue,line width=1pt]
  table[row sep=crcr]{%
5   792\\
6   1716\\
7   3432\\
8   6435\\
9   11440\\
10  19448\\
11  31824\\
12  50388\\
13  77520\\
14  116280\\
15  170544\\
16  245157\\
17  346104\\
18  480700\\
19  657800\\
20  888030\\
21  1184040\\
22  1560780\\
23  2035800\\
24  2629575\\
25  3365856\\
26  4272048\\
27  5379616\\
28  6724520\\
29  8347680\\
30  10295472\\
31  12620256\\
32  15380937\\
33  18643560\\
34  22481940\\
35  26978328\\
36  32224114\\
37  38320568\\
};

\addplot [color=teal,line width=1pt]
  table[row sep=crcr]{%
5         190\\
6         386\\
7         734\\
8        1324\\
9        2284\\
10        3790\\
11        6080\\
12        9473\\
13       14378\\
14       21323\\
15       30974\\
16       44159\\
17       61898\\
18       85440\\
19      116286\\
20      156240\\
21      207446\\
22      272432\\
23      354162\\
24      456097\\
25      582238\\
26      737205\\
27      926298\\
28     1155567\\
29     1431892\\
30     1763074\\
31     2157904\\
32     2626276\\
33     3179278\\
34     3829294\\
35     4590118\\
36     5477081\\
37     6507152\\
};
\end{axis}
\end{tikzpicture}

\ref{leg}
    \caption{The cardinality of the search space for binary matrices using different methods: Left: $\gamma_c=2$ and Right: $\gamma_c=3$.}
    \label{fig:my_label}
    \vspace{-0.8cm}
\end{figure}
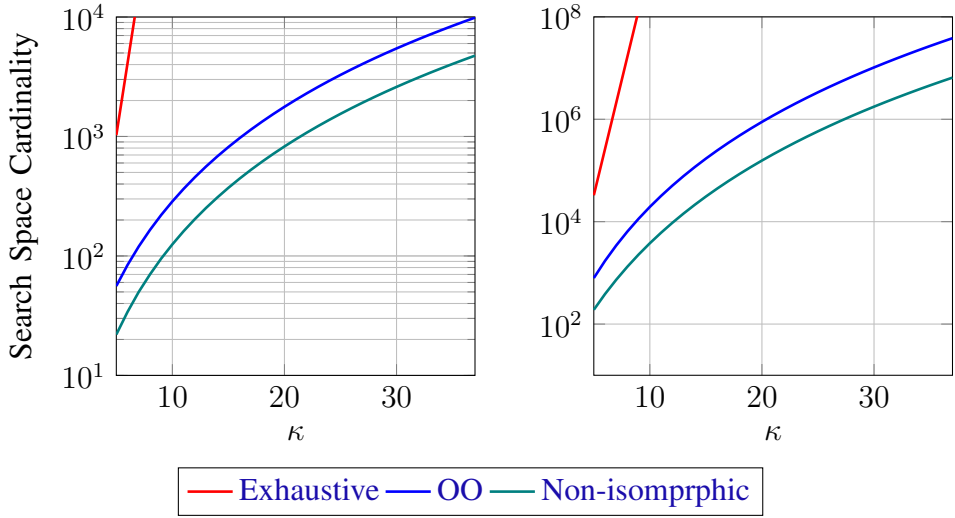
Fig.~\ref{fig:search_space_size_final} shows the cardinality of the search space as a function of $\kappa$ for $\gamma=2$ (left panel) and $\gamma=3$ (right panel), using three different schemes. As we see for the exhaustive search scheme, the cardinality, i.e., $2^{\gamma\kappa}$, grows exponentially with $\kappa$ and quickly goes beyond the practical feasibility. However, the cardinality of the search space for the OO and our new scheme based on nonequivalent matrices remain far below the one for the exhaustive search due to their smaller growth rate with $\kappa$. 
Furthermore, our scheme has smaller cardinality compared to the existing setting (OO) by almost half an order of magnitude for $\gamma=2$ and almost an order of magnitude for $\gamma=3$.\vspace{-0.4cm}
\end{example}

We finally note that evaluating each option in the search space, particularly identifying its decoding threshold, is computationally heavy and any reduction in the search space results in a  reduction of the complexity of the design algorithm. Thus, the observed reduction of an order of magnitude in the search space in Example~\ref{example:searchspace} notably reduces the design expenses.\vspace{-0.3cm}

\section{An algorithm For Joint Finite-Length Asymptotic Design of SC Codes}\vspace{-0.1cm}
\label{Sec:CC_opt}
In this section, we focus on binary partitioning matrices, i.e., $\mathbf{P}\in\{0,1\}^{\gamma\times\kappa}$ corresponding to regular SC-LDPC codes with memory $1$, i.e., $\mathbf{B}_0+\mathbf{B}_1=\mathbf{1}_{\gamma\times\kappa}$. However, the derivations in this section can be generalized to an arbitrary memory $m\geq1$ and irregular designs, i.e., $\mathbf{P}\in\{{\star},0,\dots,m\}^{\gamma\times\kappa}$. 
In the generalized case, the reduced search space in the previous section follows similarly by noting that the stars-and-bars method is considered with $(m+1)^{\gamma}$ bins, or $(m+2)^{\gamma}$ bins in case of irregular design, rather than $2^\gamma$ bins. In what follows, given $\gamma$ and $\kappa$, we produce a (short) list of partitioning matrices that offer a meaningful trade-off between threshold and cycle population. By meaningful we mean that no member of this list results in an SC protograph that is inferior to any other one (in the entire search space) in both the threshold and cycle-count properties.

\subsection{Reduced Search Space of SC-LDPC Codes}

We remind that the proto-matrix $\mathbf{B}_\text{SC}$ of a regular SC code with memory $m=1$ and coupling length $l$, in our construction, is obtained by partitioning the proto-matrix of $\mathbf{B}=\mathbf{1}_{\gamma\times\kappa}$ into $\mathbf{B}_0$ and $\mathbf{B}_1=\mathbf{B}-\mathbf{B}_0$, and coupling $l$ copies of $\mathbf{B}_0=[u_{i,j}]$ and $\mathbf{B}_1=[v_{i,j}]$ in a diagonal structure. We introduced the partitioning matrix $\mathbf{P}=[p_{i,j}]$ with size $\gamma\times\kappa$ and elements in $\{0,1,\star\}$ that fully characterizes the construction of a regular/irregular SC code with memory $1$ as follows: If $p_{i,j}=0$, $u_{i,j}=1$ and $v_{i,j}=0$; If $p_{i,j}=1$, $u_{i,j}=0$ and $v_{i,j}=1$; If $p_{i,j}=\star$, $u_{i,j}=v_{i,j}=0$.

Lemma~\ref{lemma:sc_perms} derives the congruence between coupled proto-matrices and the partitioning matrices that are used to construct them. This congruence allows searching for a coupled SC-LDPC code over a reduced search space of small ($\gamma\times \kappa$) nonequivalent partitioning matrices.\vspace{-0.3cm}

\begin{lemma}\label{lemma:sc_perms}
Any column/row permutation of the partitioning matrix $\mathbf{P}$ of an SC code results in an SC proto-matrix that is a column/row permuted version of the original SC proto-matrix $\mathbf{B}_\textrm{SC}$.\vspace{-0.3cm}
\end{lemma}

\begin{proof}
By definition, any row and column permutations on $\mathbf{P}$ automatically applies to both $\mathbf{B}_0$ and $\mathbf{B}_1$. 
This means that when $\mathbf{P}\xrightarrow{{\overline{\pi}_c,\overline{\pi}_r}}\mathbf{P}'$, we have $\mathbf{B}_0\xrightarrow{{\overline{\pi}_c,\overline{\pi}_r}}\mathbf{B}'_0$ and $\mathbf{B}_1\xrightarrow{{\overline{\pi}_c,\overline{\pi}_r}}\mathbf{B}'_1$. 
Thus, the matrix $[\mathbf{B}'_1\text{ }\mathbf{B}'_0]$ is a row permuted version of $[\mathbf{B}_1\text{ }\mathbf{B}_0]$ using $\overline{\pi}_r$, 
and the matrix $[{\mathbf{B}'_0};{\mathbf{B}'_1}]$ is a column permuted version of $[{\mathbf{B}_0};{\mathbf{B}_1}]$ using $\overline{\pi}_c$. In view of the diagonal structure of $\mathbf{B}_\text{SC}$, we can infer that $\mathbf{B}'_\text{SC}$ which has $\mathbf{B}'_0$ and $\mathbf{B}'_1$ as component matrices is row and column permuted version of $\mathbf{B}_\text{SC}$, with column permutation
$[\overline{\pi}_c,\overline{\pi}_c+\kappa,\dots,\overline{\pi}_c+(l-1)\kappa]$ and row permutation $[\overline{\pi}_r,\overline{\pi}_r+\gamma,\dots,\overline{\pi}_r+l\gamma]$, where the addition is performed element-wise.
\end{proof}

%Therefore, Lemma~\ref{lemma:sc_perms} links the results in the previous section (see Definition~\ref{Def: equivalenc}) and the design task in this section, and its main implication is that when designing an SC protograph, one can significantly reduce the complexity by only considering nonequivalent partitioning matrices. 

In Lemma~\ref{col_perm} in the appendix, we show that the state-of-the-art approach of constructing SC-LDPC codes with overlap parameters \cite{EsfahanizadehTCOM2018} results in a search space with similar size as the reduced search space obtained by just eliminating the column-wise equivalent options. Unlike the representation of column-wise nonequivalent matrices introduced in this paper, it is computationally difficult to iterate over all possible overlap parameters to find the best ones due to their dependencies. Besides, \cite{EsfahanizadehTCOM2018} does not consider the row-wise equivalence to further reduce the search space and thus results in a higher-computational complexity of the code design compared to the introduced approach in this paper.

%In fact, all of the members of the family of regular SC codes with $m=1$ and identical sets of overlap parameters are column permuted versions of each other. 
%However, the OO approach does not take into account the row permutations, i.e., two equivalent partitioning matrices may have different overlap parameters, thus the set of partitioning matrices with distinct overlap parameters is a superset of the nonequivalent partitioning matrices.

\subsection{Cycle Enumeration in the Reduced Search Space\label{sub:cycleenum}}

In this paper, we focus on cycles-$6$ as the shortest and most problematic cycles for practical LDPC codes. However, the approach presented in this subsection can be extended to longer cycles with some modifications. Let $\mathbf{B}$ be a binary matrix with size $\gamma\times \kappa$. A sequence of index pairs $[i_1,j_1]-[i_1,j_2]-[i_2,j_2]-[i_2,j_3]-[i_3,j_3]-[i_3,j_1]$, where $[i_1,i_2,i_3]\in\{1,\dots,\gamma\}^3$, $[j_1,j_2,j_3]\in\{1,\dots,\kappa\}^3$, $i_1\neq i_2\neq i_3$, and $j_1\neq j_2\neq j_3$, represents a cycle-$6$ in $\mathbf{B}=[b_{i,j}]$ iff,
\begin{equation}\label{Eq:protoCylce}
    b_{i_1,j_1}=b_{i_1,j_2}=b_{i_2,j_2}=b_{i_2,j_3}=b_{i_3,j_3}=b_{i_3,j_1}=1.
\end{equation}
Moreover, this cycle results in $z$ cycles in the lifted matrix $\mathbf{H}$ according to the power matrix $\mathbf{C}=[c_{i,j}]$ (see Section~\ref{Sub:SC-LDPCL}) iff \cite{6691250,8984289},
\begin{equation}\label{Eq:liftedCylce}
    c_{i_1,j_1}-c_{i_1,j_2}+c_{i_2,j_2}-c_{i_2,j_3}+c_{i_3,j_3}-c_{i_3,j_1}=0\quad(\mathrm{mod}\;z).
\end{equation}
In general, identifying cycles-$6$ can be done via brute-force methods. As a result, the computational complexity can be high for practical code parameters. In what follows, we exploit the structure of the SC matrix, i.e., circulant-based and repetitive structure, to reduce the complexity.

Due to the structure of SC codes with memory $m=1$, any cycle-$6$ spans either one or two replicas \cite{EsfahanizadehTCOM2018}. 
Denote the first replica and the first two replicas of $\mathbf{H}_\text{SC}$ as $\mathbf{R}_1$ and $\mathbf{R}_2$, respectively. Similarly, denote the first replica and the first two replicas of $\mathbf{B}_\text{SC}$ as $\mathbf{Q}_1$ and $\mathbf{Q}_2$, respectively. Note that the corresponding power matrices can be obtained using $\mathbf{C}$ by simple concatenations. Let function $F(\cdot)$ operate on a binary matrix and output its number of cycles-6. Then, due to the repetitive and diagonal structure of SC codes \cite{EsfTMAG}, 
\begin{equation}\label{cycle_SC_structure}
\begin{split}
    F(\mathbf{H}_\text{SC})=lF(\mathbf{R}_1)+(l-1)(F(\mathbf{R}_2)-2F(\mathbf{R}_1)),\\
    F(\mathbf{B}_\text{SC})=lF(\mathbf{Q}_1)+(l-1)(F(\mathbf{Q}_2)-2F(\mathbf{Q}_1)).
\end{split}
\end{equation}
Matrices $\mathbf{R}_1$ and $\mathbf{R}_2$ (resp., $\mathbf{Q}_1$ and $\mathbf{Q}_2$) are often notably smaller than $\mathbf{H}_\text{SC}$ (resp., $\mathbf{B}_\text{SC}$). Besides, their number of cycles-6 can be obtained using the overlap parameters, as described in (\ref{eq:cyclesA}). Finally, one can obtain $F(\mathbf{R}_1)$ (resp., $F(\mathbf{R}_2)$) from $F(\mathbf{Q}_1)$ (resp., $F(\mathbf{Q}_2)$) using \eqref{Eq:protoCylce} and \eqref{Eq:liftedCylce}.

\subsection{Decoding-Threshold Evaluation in the Reduced Search Space}

In order to evaluate the EXIT threshold (see Section~\ref{Sub:EXIT}), we construct the coupled protograph according to the partitioning matrix and the coupling length, and perform EXIT calculations. For increasing values of AWGN-channel parameter $\sigma$, we apply
\eqref{Eq:J VN} and \eqref{Eq:J CN} for each edge in the protograph until a threshold value is found such that all EXIT values on the VNs approach $1$. 

\subsection{An Algorithm For SC Code Design Offering A Design Trade-Off\label{sub:algo:trade-off}}

We now present our code-design algorithm. The inputs of the algorithm are the code parameters $\gamma,\kappa,z,l$, and the output is a list of protographs such that no member in this list is inferior to any other protograph in the entire search space (we say that protograph $\mathcal{G}_1$ is inferior to protograph $\mathcal{G}_2$ if $\mathcal{G}_1$ has both lower threshold and larger number of cycles-6 than $\mathcal{G}_2$ in the corresponding lifted graph). This candidate list is often very short, and it is sorted such that its first member is the protograph with the best (lowest) cycle-count and the worst (lowest) threshold, and the last member has the best threshold and worst (highest) number of cycles. 

The algorithm consists of the following steps:
\begin{enumerate}
    \item Generate a list of nonequivalent partitioning matrices of size $\gamma\times \kappa$ (which according to Lemma~\ref{lemma:sc_perms} corresponds to a list of nonequivalent SC proto-matrices).     %We omit the technical details on how this is done for brevity.
    \item \label{item:lift} Calculate the number of cycles-6 in the lifted coupled graphs corresponding to each partitioning matrix using \eqref{Eq:protoCylce} and \eqref{Eq:liftedCylce}.%according to Algorithm~\ref{alg:cycles}. Note that counting the cycles in the lifted graph yields a better predictor for the high-SNR BER performance than counting them in the protograph.
    \item Sort the list in ascending order according to the number of cycles-6.%(first member is the best in terms of cycles).%Note that the list still consist of partitioning matrices (i.e., and not lifted SC graphs).
    \item \label{step:exit} Iterate over the sorted list and for each partitioning matrix, generate the coupled protograph and calculate its EXIT threshold as described in Section~\ref{Sub:EXIT}. 
    \item Filter the list by removing inferior partitioning matrices using the following method:
    \begin{itemize}
        \item Initialize the final list of candidates to be empty and $\sigma^*=0$ ($\sigma^*$ records the highest found threshold).
        \item Iterate, in order, over the sorted list. If a member has a higher threshold than $\sigma^*$, append the partitioning matrix to the final candidate-list and update $\sigma^*$.  
    \end{itemize}
    %Return the final-candidate list. 
\end{enumerate}

%\begin{remark}\label{remark:trim}
%If step~\ref{step:exit} is computationally intensive, one can modify it by trimming the end of the list.
%\end{remark}
\begin{remark}\label{remark:lifting}
Although in this work we use CB lifting, one can easily use any other lifting method while keeping the general structure of the algorithm. For example, one can perform the cycle optimization of step~\ref{item:lift} over the protograph (to obtain the minimum number of cycles-6 in the protograph) and then later use a lifting optimization program as in \cite{EsfahanizadehTCOM2018}.
\end{remark}

The output of the above algorithm is a candidate list whose first member represents a choice that has the best cycle-count properties in the list, called cycle-driven (CD) choice, and the last member has the best threshold properties in the list, called threshold-driven (TD) choice. %\ycc{suggest to remove the following sentence: }In the next section, we propose how to extend this joint cycles-threshold optimization to SC-LDPC codes with sub-block locality. 

\section{Extension to SC-LDPC Codes with Sub-Block Locality}
\label{Sec:localglobal_des}

In this section, we expand our framework of designing SC-LDPC codes with jointly optimized finite-length and asymptotic performance to also incorporate the sub-block locality feature. The first and second subsections are dedicated to designing CCNs and designing LCNs, respectively.

\subsection{Global Design}
In this subsection, we discuss the design of CCNs, i.e., the entries in the first $\gamma_c$ rows of $\mathbf{P}$, in order to reduce the population of short cycles and increase the global decoding threshold. Joint cycle and threshold optimization of SC codes with no locality was studied in Section~\ref{Sec:CC_opt}; as we will see, adding locality {introduces additional opportunities for reducing the design complexity}:
\begin{itemize}
\item \textbf{Threshold Estimation per Partitioning Candidate:} a na\"ive approach for threshold optimization, is for every partitioning matrix in the set of nonequivalent options to construct the global code, i.e., $l$ replicas of $ [\mathbf{B}_0;\mathbf{B}_1]$ (see Section~\ref{Sub:SC-LDPCL}), and calculate its global EXIT threshold $ \sigma^*_G $ (see Section~\ref{Sub:EXIT}). However, this calculation can be computationally intensive since the coupled protograph is large (has many edges). As a result, instead of calculating the global threshold of the coupled protograph, we calculate a lower bound which appears implicitly in \cite{IyengarPapaleo12} as $\sigma^*_G \geq \sigma^*\left(\mathbf{B}_0\right)$, where $\sigma^*\left(\mathbf{B}_0\right)$ is the threshold of $\mathbf{B}_0$.
The reason this bound is used for the design of SC-LDPC codes with sub-block locality (i.e., not in the design in Section~\ref{Sec:CC_opt}) is that it gets tighter as $\gamma$ increases, as is the case for SC codes with sub-block locality, and offers a good approximation to the global threshold. For codes without sub-block locality, the lower bound holds, however, it is loose and sometimes even trivial, see \cite{EshedTIT}.
The complexity of calculating $\sigma^*\left(\mathbf{B}_0\right)$ is significantly smaller than the complexity of calculating $\sigma^*_G$, and this reduction enables a more affordable joint cycles-threshold optimization.
 \item \textbf{Cycle Enumeration per Partitioning Candidate:} For every nonequivalent partitioning matrix ({comprising both CCNs and LCNs}), we calculate the number of cycles-$6$ in the lifted global code, as described in Section~\ref{Sec:CC_opt}. 
 
\end{itemize}
{However, in the case of codes with sub-block locality, the structure of the codes provides further simplifications. Recall from Section~\ref{Sub:SC-LDPCL} that the matrix $\mathbf{P}$ has $\gamma_l$ rows that do not mix $0$ and $1$ entries. This simplifies the construction algorithm in two ways: 1) there are fewer nonequivalent $\mathbf{P}$ matrices to search over compared to unstructured matrices with $\gamma_c+\gamma_l$ rows, and 2) for each candidate matrix $\mathbf{P}$, calculating the number of cycles-6 it induces (using~\eqref{cycle_SC_structure}) is simplified because the function $F(\cdot)$ only needs to track overlaps among the $\gamma_c$ CCN rows, and the cycles of the full code can be shown to be fully determined by these overlaps and the power matrix when the $\gamma_l$ rows are all-$0$ (corresponding to regular local code).}
{We next quantify the reduction of the search space due to the constraint that $\mathbf{P}$ has $\gamma_c$ CCNs and $\gamma_l$ LCNs. In particular, when traversing the space of nonequivalent matrices with $\gamma_c$ rows, we exclude those} that result in an all-zero or all-one rows, as these choices effectively add one or more LCNs and result in fewer than $ \gamma_c$ CCNs. We theoretically derive the impact of this reduction for $\gamma_c=2$ and $\gamma_c=3$ in the next propositions, whose proofs appear in the appendix.

\begin{proposition}\label{prop:gamma2searcharea} Let $\gamma_c=2$ and $\kappa$ be a positive integer. The set of column distributions for all $2\times\kappa$ nonequivalent  matrices that do not have all-zero or all-one rows has cardinality $|\mathcal{K}_{\kappa,2}|-2\kappa-1$, where $|\mathcal{K}_{\kappa,2}|$ is given in Theorem~\ref{theorem:gamma2}.
\end{proposition}
% \begin{example}\label{Ex:gamma2searcharea}
% 	Let $ \kappa=11$ and $\gamma=2$. Then, the search space contains $22 - 5 = 17{\color{red}(22- 10 +1=13)}$ protographs, which is $ 15 {\color{red(20)}$ times smaller then the entire space $ \{0,1\}^{8} $.
% \end{example}

%Let $ \kappa=11$ and $\gamma=3 $. Then, there are $ \mathcal K_{11,3}= 60+1,452+4,568=6080 $ nonequivalent binary matrices with $ \gamma $ rows and $ \kappa  $ columns. This is $ 1.4128\cdot 10^6 $ times smaller then the entire space  $\{0,1\}^{\gamma \kappa }$. 

\begin{proposition}\label{prop:gamma3searcharea} Let $\gamma=3$ and $\kappa$ be a positive integer. The set of column distributions for all $3\times\kappa$ nonequivalent  matrices that do not have all-zero or all-one rows has cardinality $|\mathcal{K}_{\kappa,3}|-2|\mathcal{K}_{\kappa,2}|+(\kappa+1)$, where $|\mathcal{K}_{\kappa,2}|$ and $|\mathcal{K}_{\kappa,3}|$ are given in Theorem~\ref{theorem:gamma2} and Theorem~\ref{theorem:gamma3}, respectively.\vspace{-0.4cm}
\end{proposition}

\begin{example}\label{Ex:gamma3_1}
Let $ \kappa=11$ and $\gamma=3 $. Then, there are $6{,}080$ nonequivalent binary matrices with size $3\times11$. Among them, there are $5{,}686$ matrices that do not have all-zero or all-one rows.\vspace{-0.2cm}
\end{example}

After obtaining the reduced space of candidates for matrix $\mathbf{P}_C$, we propose using the approach discussed in Section~\ref{sub:algo:trade-off} to identify the ones that offer a meaningful trade-off between finite-length and asymptotic performance.%\ycc{suggest to remove the following sentence: }Next, we will show how to replace the local code with an irregular code. \vspace{-0.4cm}

\subsection{Local Design\label{Sec:local_des}}

The sub-block locality feature lets smaller contiguous parts of the long codewords to be decoded individually, using a local decoder and thus results in an improved decoding latency. In this sub-section, we investigate how to add irregularity into the local codes in order to improve both their decoding threshold and their cycle properties. In particular, we propose two protograph constructions for the local code of an SC-LDPCL code with parameters $\gamma_l$, $\kappa$, and $\nu$, where $\nu\in\{0,\dots,\kappa-1\}$ is the number of $\star$ elements in the $\gamma_l\times\kappa$ matrix $\mathbf{P}_L$. The two designs we propose have the same rate but may differ considerably in both their threshold and cycles properties. 

We first define some matrices that are used in the constructions. For integers $\alpha, \beta$, and $k$, let $\mathbf{Q}(\alpha,\beta;k)=[q_{i,j}]$ and $\mathbf{S}(\alpha,\beta)=[s_{i,j}]$ be $\alpha\times\beta$  matrices, such that 
\begin{equation*}
    q_{i,j} = \begin{cases}
        0&i=k  \\
        1&\text{otherwise}
    \end{cases},\quad\quad
    s_{i,j}=\begin{cases}
        0&j=i\leq\min(\alpha,\beta)  \\
        1& \text{otherwise.}
    \end{cases}
\end{equation*}
Let write $\nu=a\gamma_l+b$ with positive integers $a$ and $b$ such that $b < \gamma_l$. The \textit{balanced} and \textit{unbalanced} local-code constructions are represented by the proto-matrices $\mathbf{B}_B$ and $\mathbf{B}_U$, respectively,% and defined as follows:
\begin{flalign}\label{def_mat_BB}
\mathbf{B}_B\hspace{-0.08cm}=
{\small\left[\!\!
\begin{array}{c:c:c:c:c}
\!\mathbf{1}_{\gamma_l\times(\kappa-\nu)}\hspace{-0.08cm}&\hspace{-0.08cm} \mathbf{S}(\gamma_l,b)\hspace{-0.08cm}&\hspace{-0.08cm} \mathbf{Q}(\gamma_l,a;\gamma_l)\hspace{-0.08cm}&\hspace{-0.08cm}\dots\hspace{-0.08cm}&\hspace{-0.08cm} \mathbf{Q}(\gamma_l,a;1)\!
\end{array}\!\! \right]},\quad\mathbf{B}_U={\small\left[\!\!
\begin{array}{c:c}
\mathbf{1}(\gamma_l,\kappa-\nu) & \mathbf{Q}(\gamma_l,\nu;1) 
\end{array}\!\! \right],}
\end{flalign}
where the vertical dashed lines represent the horizontal concatenation of sub-matrices. $\mathbf{B}_B$ and $\mathbf{B}_U$ are both $\gamma_l \times \kappa$ matrices with $\nu$ zero entries; in $\mathbf{B}_B$, zeros are uniformly distributed among the rows, while in $\mathbf{B}_U$, all zeros are in the first row.\vspace{-0.3cm}

\begin{example}\label{Ex:Const 1}
	Let $ \gamma_l=3 $, $\kappa=13$, and $ \nu = 10 $. Then,
	\begin{align*}
	&\mathbf{B}_B = 
	\left[\begin{array}{c:c:c:c:c}
	1\,1\,1&0&1\,1\,1&1\,1\,1&0\,0\,0\\
	1\,1\,1&1&1\,1\,1&0\,0\,0&1\,1\,1\\
	1\,1\,1&1&0\,0\,0&1\,1\,1&1\,1\,1
	\end{array}\right],\quad\mathbf{B}_U = 
	\left[\begin{array}{c:cccc}
    1\,1\,1&0&0\,0\,0&0\,0\,0&0\,0\,0 \\
    1\,1\,1&1&1\,1\,1&1\,1\,1&1\,1\,1\\
    1\,1\,1&1&1\,1\,1&1\,1\,1&1\,1\,1
    \end{array} \right].
	\end{align*}
\end{example}

%Given two protographs, it is not clear, in general, which one yields a better threshold since many parameters need to be tracked. However, in some cases we can order the thresholds of two protographs. The following ordering of regular protographs (with different rates) will be used in the sequel as a supporting lemma.

\begin{proposition}\label{Prop:thU<thB}
	Let $\kappa$, $\gamma_l$, and $\nu<\kappa$ be positive integers. If $ \kappa-\lfloor\tfrac{\nu}{\gamma_l}\rfloor \leq \nu$, then $ \sigma^*(\mathbf{B}_U)\leq  \sigma^*(\mathbf{B}_B) $.
\end{proposition}

Next, we investigate the cycle properties of the balanced and unbalanced local constructions.

\begin{proposition}\label{Prop:C6U<C6B}
	 Let $\gamma_l\!=\!3$, $\kappa\!>\!0$, and $\nu\!=\!a\gamma_l\!+\!b\!<\!\kappa$, and let $F(\mathbf{B}_B)$ and $F(\mathbf{B}_U)$ denote the number of cycles-$6$ in the protograph of the balanced and unbalanced local codes, respectively. Then $F(\mathbf{B}_U)\leq  F(\mathbf{B}_B)$. 
\end{proposition}

\begin{proposition}\label{pro_FB_FU_gamma4} Let $\gamma_l\!=\!4$, $\kappa\!>\!0$, and $\nu\!=\!a\gamma_l<\!\kappa$, and let $F(\mathbf{B}_B)$ and $F(\mathbf{B}_U)$ denote the number of cycles-$6$ in the protograph of the balanced and unbalanced local codes, respectively. Then $F(\mathbf{B}_U)>  F(\mathbf{B}_B)$.
\end{proposition}

\begin{remark}
In Proposition~\ref{pro_FB_FU_gamma4}, we assumed $\nu$ is divisible by $\gamma_l$ only for simplicity. One can find a condition on $\nu$ for general case $\nu=a\gamma_l+b$ (where $a\geq 0$ and $0\leq b<\gamma_l$) such that $F(\mathbf{B}_U)>F(\mathbf{B}_B)$, by formulating the overlap parameters in terms of parameters $a$, $b$, and $\kappa$.
\end{remark}
\begin{remark}
For $\gamma_l=3$, Propositions~\ref{Prop:thU<thB},~\ref{Prop:C6U<C6B} show that there is a trade-off between cycle and threshold properties of these local codes, and it is the designer discretion to choose between balanced and unbalanced schemes, depending on which feature is more desirable. This trade-off does not exist for $\gamma_l=4$, where the balanced scheme has better performance in both features.\vspace{-0.3cm}
\end{remark}

We call an SC-LDPC code whose both global code and local code are designed to achieve the best cycle-count properties, resp., threshold properties, a locality-aware cycle-driven (LA-CD) choice, resp., locality-aware threshold-driven (LA-CD) choice.

\section{Simulation Results}
\label{Sec:Simulations}

In our simulations, we consider parameters $\kappa=11$, $z=67$ $\gamma_c=3$, $m=1$, $l=5$, $\gamma_l\in\{0,2,3\}$, and power matrix $C=[c_{i,j}]$ with $c_{i,j}=6{\cdot}i{\cdot}j \;(\mathrm{mod}\;z)$, which yields cycle-4 free graphs \cite{FanTurbo2000}. We investigate the performance of SC-LDPC codes with and without sub-block locality constructed using various methods (the new introduced methods and existing methods). Our results include the BER performance, cycle-counts, and threshold values. For Monte Carlo simulations, we observed at least $50$ frame errors in every reported point.\vspace{-0.4cm}

\subsection{Setup~1: Regular SC-LDPC Code Without Locality}

In this subsection, we consider $\gamma_l=0$ (SC codes with no locality), and compare four different design methods for $\mathbf{P}=[p_{i,j}]$:

\begin{itemize}
    \item\label{Item:CV_method} \textit{Cutting-vector (CV) partitioning} \cite{MitchellISIT2014}: This is partitioning via a cutting vector with size $\gamma$ whose elements $0<\zeta_1<\ldots<\zeta_{\gamma}$ are natural numbers. Then, $p_{i,j}=0$ if and only if $j<\zeta_i$. We consider the cutting vector $[4,8,11]$ for the simulations.
    \item \label{Item:Unconstrained_OO} 
\textit{Optimal overlap (OO) partitioning} \cite{EsfahanizadehTCOM2018}: The OO partitioning results in the minimum number of cycles-$6$ in the protograph SC code.
\item \label{Item:CD} 
\textit{Cycle-driven (CD) partitioning}: This is the partitioning within our reduced search space that results in the minimum number of cycles-$6$ in the \emph{lifted} graph.% (see Section~\ref{sub:algo:trade-off}).
\item \label{Item:TD} \textit{Threshold-driven (TD) partitioning}: This is the partitioning within our reduced search space that has the maximum threshold.%(see Section~\ref{sub:algo:trade-off}).
\end{itemize}

We first record the populations of cycles-$6$ along with the threshold values for the SC-LDPC codes that are constructed via the above methods. 
The results are given in the left panel of Table~\ref{Tab:GamL_02_GamC_3_cycle_pop_threshold}, where it is shown that the CD method yields $54\%$ reduction in the population of cycles-$6$ (in lifted graphs) compared to the CV method, while this reduction is less than $1\%$ for the OO method compared to the CV method. We remind that the OO method results in the minimum number of cycles-6 in the protograph, not necessarily in the lifted graph. In terms of the asymptotic behavior, the TD method results in the highest threshold while also having fewer number of cycles compared to the CV and OO methods.

\begin{table}
\centering
\caption{Cycle and threshold properties of various design methods for SC-LDPC codes with parameters $\kappa=11$, $z=67$, $m=1$, $l=5$. Left table: Without locality ($\gamma=3$); Right table: With locality ($\gamma_l=2$, $\gamma_c=5$).}
\vspace{-0.5cm}
\setlength\tabcolsep{4pt}
\begin{tabular}{l c c}
design & cycles-$6$ & $\sigma^*$\\
\hline
\hline
CV \cite{MitchellISIT2014} & $7,638$ & $0.6779$ \\
\hline
OO \cite{EsfahanizadehTCOM2018} & $7,571$ & $0.6901$ \\
\hline
CD & $3,551$ & $0.6851$ \\
\hline
TD & $5,628$ & $0.6909$ \\
\hline
\end{tabular} \quad\quad\quad
\begin{tabular}{l c c}
design & cycles-$6$ & $\sigma^*$\\
\hline
\hline
LB-CV & $83,348$ & $0.825$ \\
\hline
LB-OO & $49,714$ & $0.8433$ \\
\hline
LA-CD & $41,540$ & $0.8438$ \\
\hline
LA-TD & $49,044$ & $0.8983$ \\
\hline
\end{tabular} 
\label{Tab:GamL_02_GamC_3_cycle_pop_threshold}
\vspace{-1cm}
\end{table}

Fig.~\ref{Fig:GamL_0_GamC_3_cycle_pop_threshold} compares the BER performance for these SC-LDPC codes. The left sub-figure shows the BER performance in the low-SNR region and in particular the superiority of the TD partitioning with about half an order of magnitude compared to the CV partitioning at SNR$=2.5$~dB. The right sub-figure shows the BER performance in the high-SNR region and the superiority of the CD partitioning with about one order of magnitude compared to the CV partitioning at SNR$=5$~dB. Moreover, there is a crossover point at SNR $\simeq3$~dB, where the BER performances of CD and TD methods intersects.

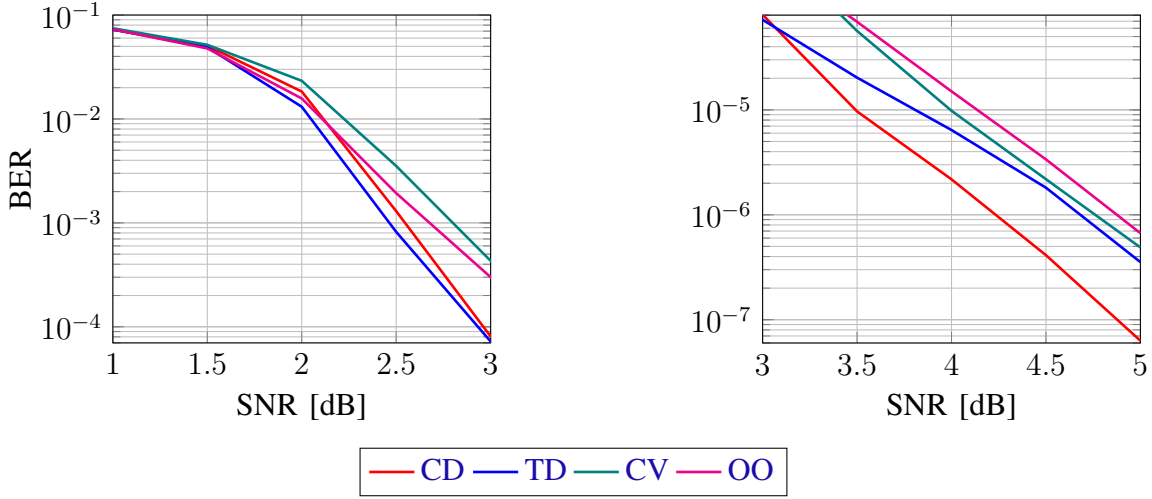
\begin{figure}
\centering
\begin{tikzpicture}

\begin{axis}[%
width=0.4\linewidth,
height=0.36\linewidth,
at={(0,0)},
xmin=1,
xmax=3,
xlabel={SNR [dB]},
ymode=log,
ymin=7e-05,
ymax=0.1,
yminorticks=true,
ylabel={BER},
xmajorgrids,
ymajorgrids,
yminorgrids,
]
\addplot [color=red,line width=1pt]
  table[row sep=crcr]{%
0	0.108959100600892\\
1	0.0722517283711314\\
1.5	0.0506900562124443\\
2	0.0183640240356658\\
2.5	0.00129625482472896\\
3	8.08088243905459e-05\\
3.5	9.67232000582644e-06\\
4	2.17715590087988e-06\\
4.5	4.12141250440116e-07\\
5	6.30931495964934e-08\\
};
%\addlegendentry{Cycles Driven}

\addplot [color=blue,line width=1pt]
  table[row sep=crcr]{%
0	0.10952316340376\\
1	0.072710473605996\\
1.5	0.0490065904245009\\
2	0.0131078374361956\\
2.5	0.000819148299842278\\
3	7.22101548441949e-05\\
3.5	2.0298113519036e-05\\
4	6.42603282716579e-06\\
4.5	1.81255238849818e-06\\
5	3.54134857565943e-07\\
};
%\addlegendentry{Threshold Driven}

\addplot [color=teal,line width=1pt]
  table[row sep=crcr]{%
1	0.074640434192673\\
1.5	0.051697034308975\\
2	0.023281643729405\\
2.5	0.003526846288040\\
3	4.312151727146053e-04\\
3.5	5.654572903588130e-05\\
4	9.808157962706255e-06\\
4.5	2.184313454678955e-06\\
5	4.884341948438173e-07\\
};
%\addlegendentry{CV}

\addplot [color=magenta,line width=1pt]
  table[row sep=crcr]{%
1	0.072910447761194\\
1.5	0.047716611746462\\
2	0.015744330296569\\
2.5	0.001931966959461\\
3	3.003067060768982e-04\\
3.5	6.912184975357486e-05\\
4	1.501276989545539e-05\\
4.5	3.372690444821332e-06\\
5	6.686301132233171e-07\\
};

\end{axis}

\begin{axis}[%
width=0.4\linewidth,
height=0.36\linewidth,
at={(3.4in,0)},
xmin=3,
xmax=5,
xlabel={SNR [dB]},
ymode=log,
ymin=6e-08,
ymax=8e-05,
yminorticks=true,
xmajorgrids,
ymajorgrids,
yminorgrids,
legend columns=4, 
legend entries={ CD, TD, CV, OO},
legend to name= leg,
legend style={
    legend cell align=left, 
    align=left, 
    draw=black,
    /tikz/column 2/.style={column sep=0.5pt},
}
]
\addplot [color=red,line width=1pt]
  table[row sep=crcr]{%
0	0.108959100600892\\
1	0.0722517283711314\\
1.5	0.0506900562124443\\
2	0.0183640240356658\\
2.5	0.00129625482472896\\
3	8.08088243905459e-05\\
3.5	9.67232000582644e-06\\
4	2.17715590087988e-06\\
4.5	4.12141250440116e-07\\
5	6.30931495964934e-08\\
};
%\addlegendentry{Cycles Driven}

\addplot [color=blue,line width=1pt]
  table[row sep=crcr]{%
0	0.10952316340376\\
1	0.072710473605996\\
1.5	0.0490065904245009\\
2	0.0131078374361956\\
2.5	0.000819148299842278\\
3	7.22101548441949e-05\\
3.5	2.0298113519036e-05\\
4	6.42603282716579e-06\\
4.5	1.81255238849818e-06\\
5	3.54134857565943e-07\\
};
%\addlegendentry{Threshold Driven}

\addplot [color=teal,line width=1pt]
  table[row sep=crcr]{%
1	0.074640434192673\\
1.5	0.051697034308975\\
2	0.023281643729405\\
2.5	0.003526846288040\\
3	4.312151727146053e-04\\
3.5	5.654572903588130e-05\\
4	9.808157962706255e-06\\
4.5	2.184313454678955e-06\\
5	4.884341948438173e-07\\
};
%\addlegendentry{CV}

\addplot [color=magenta,line width=1pt]
  table[row sep=crcr]{%
1	0.072910447761194\\
1.5	0.047716611746462\\
2	0.015744330296569\\
2.5	0.001931966959461\\
3	3.003067060768982e-04\\
3.5	6.912184975357486e-05\\
4	1.501276989545539e-05\\
4.5	3.372690444821332e-06\\
5	6.686301132233171e-07\\
};

\end{axis}

\end{tikzpicture}%

\ref{leg}
\caption{BER performance of various design methods for SC-LDPC codes with parameters $\kappa=11$, $z=67$, $\gamma=3$, $\gamma_l=0$, $m=1$, and $l=5$. Left: low-SNR region; right: high-SNR region.\vspace{-1cm}}
\label{Fig:GamL_0_GamC_3_cycle_pop_threshold}
\end{figure}

\subsection{Setup~2: Regular SC-LDPC Code with Locality}

In this subsection,  we consider $\gamma_l=2$ (SC codes with  locality) and $\mathbf{P}_L=\mathbf{0}_{\gamma_l\times\kappa}$ (last $\gamma_l$ rows of $\mathbf{P}$ are all-$0$, that is, the matrix $\mathbf{B}_0$ has $\gamma_l$ all-$1$ rows). Then, we compare four different design methods for the matrix $\mathbf{P}_C$:
\begin{itemize}
    \item\textit{Locality-blind cutting-vector (LB-CV) partitioning}: This is partitioning matrix $\mathbf{P}_C$ via cutting vector $[4,8,11]$ \cite{MitchellISIT2014}, similar to the case with no locality.
    \item\textit{Locality-blind optimal overlap (LB-OO) partitioning}: This is partitioning matrix $\mathbf{P}_C$ via OO approach \cite{EsfahanizadehTCOM2018}, ignoring the existence of the local part $\mathbf{P}_L$.
    \item\textit{Locality-aware cycle-driven (LA-CD) partitioning}: This is the optimal partitioning within our reduced search space that results in the minimum number of cycles-$6$.% in the lifted graph including the part $\mathbf{P}_L$ (see Section~\ref{sub:algo:trade-off}).
    \item\textit{Locality-aware threshold-driven (LA-TD) partitioning}: This is the optimal partitioning within our reduced search space that results in the maximum threshold.% in the lifted graph including the part $\mathbf{P}_L$ (see Section~\ref{sub:algo:trade-off}).
\end{itemize}

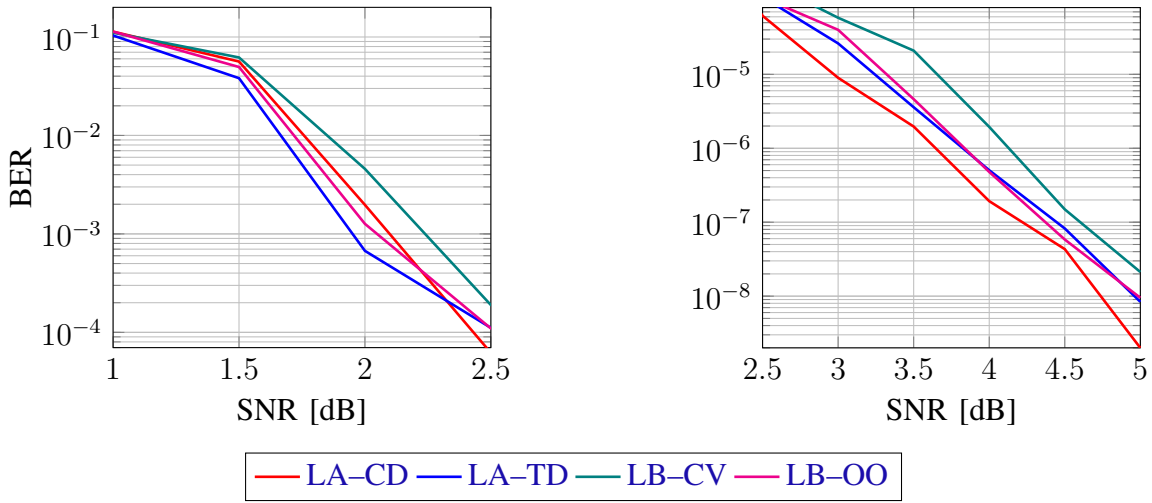
\begin{figure}
\centering
\begin{tikzpicture}

\begin{axis}[%
width=0.40\linewidth,
height=0.37\linewidth,
at={(0,0)},
xmin=1,
xmax=2.5,
xlabel={SNR [dB]},
ymode=log,
ymin=7e-05,
ymax=0.2,
yminorticks=true,
ylabel={BER},
xmajorgrids,
ymajorgrids,
yminorgrids,
legend columns=4, 
legend entries={ LA--CD, LA--TD, LB--CV, LB--OO},
legend to name= leg,
legend style={
    legend cell align=left, 
    align=left, 
    draw=black,
    /tikz/column 2/.style={column sep=0.5pt},
}
]
\addplot [color=red,line width=1pt]
  table[row sep=crcr]{%
0	0.153093622795115\\
1	0.113367577695936\\
1.5	0.0563132389998062\\
2	0.00197131226981973\\
2.5	6.19944880940786e-05\\
3	8.97677185785747e-06\\
3.5	1.96870667354629e-06\\
4	1.92864768249648e-07\\
4.5	4.32493904782837e-08\\
5	1.99536927270018e-09\\
};
%%\addlegendentry{Cycles Driven}

\addplot [color=blue,line width=1pt]
  table[row sep=crcr]{%
0	0.151649544485365\\
1	0.10353104606836\\
1.5	0.0381604962201977\\
2	0.000672702357825131\\
2.5	0.00011117579502699\\
3	2.62056056351904e-05\\
3.5	3.59203343886759e-06\\
4	5.04394233190204e-07\\
4.5	8.22586446351671e-08\\
5	8.38050781040581e-09\\
};
%\addlegendentry{Threshold Driven}

\addplot [color=teal,line width=1pt]
  table[row sep=crcr]{%
1	0.110540051584828\\
1.5	0.061927699166505\\
2	0.004571124720378\\
2.5	1.895922676634502e-04\\
3	5.826110944595089e-05\\
3.5	2.092720420143358e-05\\
4	1.946788316245458e-06\\
4.5	1.492425381529221e-07\\
5	2.119864125053369e-08\\
};
%\addlegendentry{CV}

\addplot [color=magenta,line width=1pt]
  table[row sep=crcr]{%
1	0.113331071913161\\
1.5	0.049545454545455\\
2	0.001264936315490\\
2.5	1.096176083309382e-04\\
3	3.999757678914343e-05\\
3.5	4.618146776278494e-06\\
4	4.786271862399168e-07\\
4.5	5.834434871244911e-08\\
5	9.497929104611177e-09\\
};

\end{axis}

\begin{axis}[%
width=0.40\linewidth,
height=0.37\linewidth,
at={(3.4in,0)},
xmin=2.5,
xmax=5,
xlabel={SNR [dB]},
ymode=log,
ymin=2e-09,
ymax=8e-05,
yminorticks=true,
xmajorgrids,
ymajorgrids,
yminorgrids,
legend style={
    legend cell align=left, 
    align=left, 
    draw=black
}
]
\addplot [color=red,line width=1pt]
  table[row sep=crcr]{%
0	0.153093622795115\\
1	0.113367577695936\\
1.5	0.0563132389998062\\
2	0.00197131226981973\\
2.5	6.19944880940786e-05\\
3	8.97677185785747e-06\\
3.5	1.96870667354629e-06\\
4	1.92864768249648e-07\\
4.5	4.32493904782837e-08\\
5	1.99536927270018e-09\\
};
%\addlegendentry{Cycles Driven}

\addplot [color=blue,line width=1pt]
  table[row sep=crcr]{%
0	0.151649544485365\\
1	0.10353104606836\\
1.5	0.0381604962201977\\
2	0.000672702357825131\\
2.5	0.00011117579502699\\
3	2.62056056351904e-05\\
3.5	3.59203343886759e-06\\
4	5.04394233190204e-07\\
4.5	8.22586446351671e-08\\
5	8.38050781040581e-09\\
};
%\addlegendentry{Threshold Driven}

\addplot [color=teal,line width=1pt]
  table[row sep=crcr]{%
1	0.110540051584828\\
1.5	0.061927699166505\\
2	0.004571124720378\\
2.5	1.895922676634502e-04\\
3	5.826110944595089e-05\\
3.5	2.092720420143358e-05\\
4	1.946788316245458e-06\\
4.5	1.492425381529221e-07\\
5	2.119864125053369e-08\\
};
%\addlegendentry{CV}

\addplot [color=magenta,line width=1pt]
  table[row sep=crcr]{%
1	0.113331071913161\\
1.5	0.049545454545455\\
2	0.001264936315490\\
2.5	1.096176083309382e-04\\
3	3.999757678914343e-05\\
3.5	4.618146776278494e-06\\
4	4.786271862399168e-07\\
4.5	5.834434871244911e-08\\
5	9.497929104611177e-09\\
};

\end{axis}

\end{tikzpicture}%

\ref{leg}
\caption{BER performance of various design methods for SC-LDPC codes with sub-block locality with parameters $\kappa=11$, $z=67$, $\gamma_c=3$, $\gamma_l=2$, $m=1$, $l=5$, and $P_L=\mathbf{0}$. Left: low-SNR region; right: high-SNR region.\label{Fig:BER_Loc}}
\vspace{-0.8cm}
\end{figure}

The results for population of cycles-$6$ in the lifted graphs along with the thresholds for SC-LDPC codes constructed with above methods are given in the right panel of Table~\ref{Tab:GamL_02_GamC_3_cycle_pop_threshold}. As we see, the LA-CD method yields about $50\%$ reduction in the population of cycles-$6$  compared to the LB-CV method, while this reduction is about $40\%$ for the LB-OO method compared to the LB-CV method. In terms of the asymptotic behavior, the LA-TD method results in the highest threshold as motivated by the design while also having fewer number of cycles in the lifted graph compared to the LB-CV and LB-OO methods. As for the BER performance, Fig.~\ref{Fig:BER_Loc} shows the results for SC-LDPC codes with sub-block locality according to the above constructions. Again, the left and right sub-figures show the performance in the low-SNR region and high-SNR region, respectively. Consistent with the design goals, the LA-TD design has superior performance in the low-SNR region, i.e., about $0.7$ of an order of magnitude compared to the LB-CV method at SNR$=2$~dB, while the LA-CD design has superior performance in the high-SNR region, i.e., about one order of magnitude compared to the LB-CV method at SNR$=5$~dB. Moreover, there is a crossover point at SNR $\simeq2.3$~dB, where the BER performances of LA-CD and LA-TD methods intersects.\vspace{-0.4cm}

\subsection{Setup 3: Irregular SC-LDPC Code with Locality}

In this subsection, we consider $\gamma_l= 3$ (SC codes with locality), $\kappa=11, \nu=8$, and $\mathbf{P}_L$ constructed according to balanced and unbalanced designs that were presented in Section~\ref{Sec:local_des}. 
In fact, we consider the balanced design for the local code when the LA-TD method is used for the global design, as the balanced design was shown to have superior threshold performance (see Proposition~\ref{Prop:thU<thB}). 
Similarly, we consider the unbalanced design for the local code when the LA-CD method is used for the global design, as the unbalanced design was shown to have lower cycle-$6$ count for $\gamma=3$ (see Proposition~\ref{Prop:C6U<C6B}). 
We then investigate the performance of local decoding and global decoding for these construction methods (the performance of local decoding corresponds to the performance of the local LDPC code).

Table~\ref{Tab:GamL_3_GamC_3_cycle_pop_threshold} shows the population of cycles-$6$ in the lifted graphs along with the threshold values for the regular and irregular SC-LDPC codes with sub-block locality. As seen, adding unbalanced irregularity results in $63\%$ reduction in the population of cycles-$6$  in the cycle-driven approach, and adding balanced irregularity further improves the decoding threshold in the threshold-driven approach. For the local codes, there is a trade-off between threshold and cycle properties of the balanced and unbalanced schemes, as expected.

\begin{table}
\centering
\caption{Cycle and threshold properties of cycle-driven and threshold-driven methods for SC-LDPC codes with sub-block locality and parameters $\kappa=11$, $z=67$, $\gamma_c=3$, $\gamma_l=3$, $m=1$, and $l=5$.\vspace{-0.4cm}}
\setlength\tabcolsep{4pt}
\begin{tabular}{c|c|c|ccc}
& global design & local design & cycles-$6$ & $\sigma^*$\\
\hline
\hline
\multirow{4}{*}{\rotatebox[origin=c]{90}{global code}}& \multirow{2}{*}{LA-CD} & regular & $89,847$ & $0.8805$ \\
\cline{3-6}
& & unbalanced & $33,031$ & $0.9245$ \\
\cline{2-6}
& \multirow{2}{*}{LA-TD} & regular & $137,082$ & $0.9568$ \\
\cline{3-6}
&  & balanced & $34,170$ & $0.9867$ \\
\hline
\hline
\multirow{2}{*}{\rotatebox[origin=c]{90}{local code}} & \multirow{2}{*}{-} & unbalanced & $268$ & $0.5271$ \\
\cline{3-6}
&  & balanced & $ 536 $ & $0.5979$\vspace{0.1cm} \\
\hline
\end{tabular} 
\label{Tab:GamL_3_GamC_3_cycle_pop_threshold}
\vspace{-0.8cm}
\end{table}

Fig.~\ref{Fig:BER_LocRegIrreg} compares the BER performance of the mentioned constructions using global decoding, with the left sub-figure for low-SNR region and right one for high-SNR region. As shown, adding irregularity improves the performance of the threshold driven design by $0.8$ order of magnitude at SNR$=2$~dB. Moreover, adding irregularity improves the performance of the cycle driven design by approximately one order of magnitude at SNR$=4$~dB. Fig.~\ref{Fig:LocalDec} shows the local-decoding BER performance of the balanced and unbalanced schemes corresponding to the local part of the codes in Fig.~\ref{Fig:BER_LocRegIrreg}. As seen, the BER curves exemplify the theoretical results.

\begin{figure}
\centering
\begin{tikzpicture}

\begin{axis}[%
width=0.40\linewidth,
height=0.37\linewidth,
at={(0,0)},
xmin=1,
xmax=2.5,
xlabel={SNR [dB]},
ymode=log,
ymin=7e-05,
ymax=0.2,
yminorticks=true,
ylabel={BER},
xmajorgrids,
ymajorgrids,
yminorgrids,
legend columns=4, 
legend entries={ LA-CD Regular, LA-TD Regular, LA-CD  Unbalanced, LA-TD  Balanced},
legend to name= legg,
legend style={
    legend cell align=left, 
    align=left, 
    draw=black,
    /tikz/column 2/.style={column sep=0.5pt},
}
]
\addplot [color=red,line width=1pt]
  table[row sep=crcr]{%
1	0.148569490211281\\
1.5	0.125518511339407\\
2	0.072225237449118\\
2.5	0.007047877495639\\
3	1.415774510493917e-04\\
3.5	2.275745910513250e-05\\
4	1.071646971775687e-06\\
4.5	7.734018317696742e-08\\
5	8.819516621493385e-09\\
};

\addplot [color=blue,line width=1pt]
  table[row sep=crcr]{%
1	0.139125799573561\\
1.5	0.090527233960070\\
2	0.019119984493119\\
2.5	0.001973730224230\\
3	3.065946606521876e-04\\
3.5	5.209342895910060e-05\\
4	1.185937570851050e-05\\
4.5	1.173973308454419e-06\\
5	1.370413772083107e-07\\
};

\addplot [color=red,dashed,line width=1pt]
  table[row sep=crcr]{%
1	0.131316146540027\\
1.5	0.085090133746850\\
2	0.011670866446986\\
2.5	1.767111269106245e-04\\
3	1.829636263698248e-05\\
3.5	1.756327957219025e-06\\
4	1.417903358244403e-07\\
4.5	1.560372116728290e-08\\
5	7.462649253917910e-09\\
};

\addplot [color=blue,dashed,line width=1pt]
  table[row sep=crcr]{%
1	0.126586547780578\\
1.5	0.053266136848226\\
2	0.003171157201008\\
2.5	4.482830602233588e-04\\
3	8.297160125675333e-05\\
3.5	1.351462342587820e-05\\
4	1.698315078740303e-06\\
4.5	3.025765061133989e-07\\
5	2.374479308064790e-08\\
};

\end{axis}

\begin{axis}[%
width=0.40\linewidth,
height=0.37\linewidth,
at={(3.4in,0)},
xmin=2.5,
xmax=4.5,
xlabel={SNR [dB]},
ymode=log,
ymin=2e-09,
ymax=8e-05,
yminorticks=true,
xmajorgrids,
ymajorgrids,
yminorgrids,
legend style={
    legend cell align=left, 
    align=left, 
    draw=black
}
]
\addplot [color=red,line width=1pt]
  table[row sep=crcr]{%
1	0.148569490211281\\
1.5	0.125518511339407\\
2	0.072225237449118\\
2.5	0.007047877495639\\
3	1.415774510493917e-04\\
3.5	2.275745910513250e-05\\
4	1.071646971775687e-06\\
4.5	7.734018317696742e-08\\
5	8.819516621493385e-09\\
};

\addplot [color=blue,line width=1pt]
  table[row sep=crcr]{%
1	0.139125799573561\\
1.5	0.090527233960070\\
2	0.019119984493119\\
2.5	0.001973730224230\\
3	3.065946606521876e-04\\
3.5	5.209342895910060e-05\\
4	1.185937570851050e-05\\
4.5	1.173973308454419e-06\\
5	1.370413772083107e-07\\
};

\addplot [color=red,dashed,line width=1pt]
  table[row sep=crcr]{%
1	0.131316146540027\\
1.5	0.085090133746850\\
2	0.011670866446986\\
2.5	1.767111269106245e-04\\
3	1.829636263698248e-05\\
3.5	1.756327957219025e-06\\
4	1.417903358244403e-07\\
4.5	1.560372116728290e-08\\
5	7.462649253917910e-09\\
};

\addplot [color=blue,dashed,line width=1pt]
  table[row sep=crcr]{%
1	0.126586547780578\\
1.5	0.053266136848226\\
2	0.003171157201008\\
2.5	4.482830602233588e-04\\
3	8.297160125675333e-05\\
3.5	1.351462342587820e-05\\
4	1.698315078740303e-06\\
4.5	3.025765061133989e-07\\
5	2.374479308064790e-08\\
};

\end{axis}

\end{tikzpicture}%

\ref{legg}
\caption{BER performance of the cycle-driven and threshold-driven methods for regular and irregular SC-LDPC codes with parameters $\kappa=11$, $z=67$, $\gamma_c=3$, $\gamma_l=3$, $m=1$, and $l=5$. Left: low-SNR region; right: high-SNR region.\label{Fig:BER_LocRegIrreg}}
\end{figure}
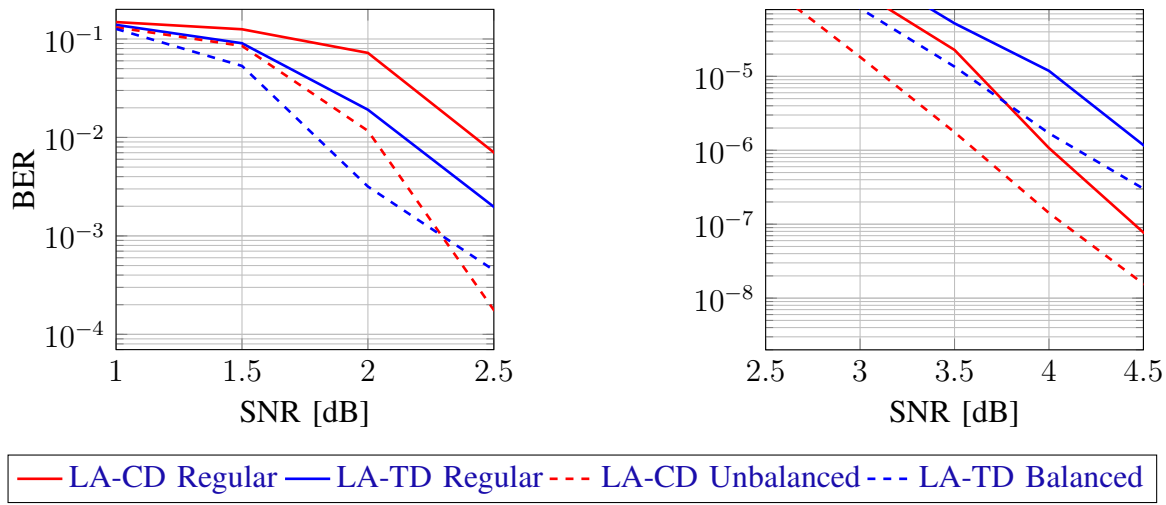

\begin{figure}
\centering
\begin{tikzpicture}

\begin{axis}[%
width=0.30\linewidth,
height=0.26\linewidth,
at={(0,0)},
scale only axis,
xmin=3,
xmax=8,
xlabel={SNR $[\mathrm{dB}]$},
ymode=log,
ymin=5e-5,
ymax=0.05,
yminorticks=true,
ylabel={BER},
xmajorgrids,
ymajorgrids,
yminorgrids,
legend columns=2, 
legend entries={ Unbalanced, Balanced},
legend to name= legend,
legend style={
    legend cell align=left, 
    align=left, 
    draw=black,
    /tikz/column 2/.style={column sep=0.5pt},
}
]

%\addplot [color=black, line width=1pt]
%  table[row sep=crcr]{%
%3.04119982655925	0.0751696065128901\\
%3.54119982655925	0.0529366156231828\\
%4.04119982655925	0.0268220585384765\\
%4.54119982655925	0.00951734832331847\\
%5.04119982655925	0.00199537073988347\\
%5.54119982655925	0.000406905203873798\\
%6.04119982655925	7.8312366100914e-05\\
%6.54119982655925	1.39716968825336e-05\\
%7.04119982655925	2.30154887217423e-06\\
%7.54119982655925	4.1383782226272e-07\\
%8.04119982655925	2.54408815302085e-08\\
%8.54119982655925	1.01763780529331e-08\\
%9.04119982655925	0\\
%9.54119982655925	0\\
%10.0411998265592	0\\
%};
%\addlegendentry{Regular}

\addplot [color=red, dashed,line width=1pt]
  table[row sep=crcr]{%
3.04119982655925	0.067949214964140\\
3.54119982655925	0.050465206435356\\
4.04119982655925	0.035617367706920\\
4.54119982655925	0.020003876720295\\
5.04119982655925	0.009808102345416\\
5.54119982655925	0.004575827594625\\
6.04119982655925	0.001994793417427\\
6.54119982655925	6.478445688091409e-04\\
7.04119982655925	2.202966661771185e-04\\
7.54119982655925	6.549364263079888e-05\\
8.04119982655925	1.694721366623310e-05\\
8.54119982655925	3.663464043785614e-06\\
9.04119982655925	8.208900497877278e-07\\
9.54119982655925	1.289003052949457e-07\\
10.0411998265592	6.784252035288755e-09\\
};
%\addlegendentry{Unbalanced}

\addplot [color=blue, dashed,line width=1pt]
  table[row sep=crcr]{%
3.04119982655925	0.066476061252181\\
3.54119982655925	0.048371777476255\\
4.04119982655925	0.026603993021903\\
4.54119982655925	0.012502422950184\\
5.04119982655925	0.003774179058169\\
5.54119982655925	0.001476196334112\\
6.04119982655925	4.919508981477619e-04\\
6.54119982655925	1.946682266015831e-04\\
7.04119982655925	6.557764172007047e-05\\
7.54119982655925	2.695688135906781e-05\\
8.04119982655925	9.317698622314360e-06\\
8.54119982655925	3.059670895828355e-06\\
9.04119982655925	9.113462536210862e-07\\
9.54119982655925	2.611927238871268e-07\\
10.0411998265592	3.052901967511872e-08\\
};
%\addlegendentry{Balanced}

\end{axis}

\begin{axis}[%
width=0.30\linewidth,
height=0.26\linewidth,
at={(3.4in,0)},
scale only axis,
xmin=8,
xmax=10,
xlabel={SNR $[\mathrm{dB}]$},
ymode=log,
ymin=6e-9,
ymax=5e-5,
yminorticks=true,
xmajorgrids,
ymajorgrids,
yminorgrids,
]

%\addplot [color=black, line width=1pt]
%  table[row sep=crcr]{%
%3.04119982655925	0.0751696065128901\\
%3.54119982655925	0.0529366156231828\\
%4.04119982655925	0.0268220585384765\\
%4.54119982655925	0.00951734832331847\\
%5.04119982655925	0.00199537073988347\\
%5.54119982655925	0.000406905203873798\\
%6.04119982655925	7.8312366100914e-05\\
%6.54119982655925	1.39716968825336e-05\\
%7.04119982655925	2.30154887217423e-06\\
%7.54119982655925	4.1383782226272e-07\\
%8.04119982655925	2.54408815302085e-08\\
%8.54119982655925	1.01763780529331e-08\\
%9.04119982655925	0\\
%9.54119982655925	0\\
%10.0411998265592	0\\
%};
%\addlegendentry{Regular}

\addplot [color=red, dashed,line width=1pt]
  table[row sep=crcr]{%
3.04119982655925	0.067949214964140\\
3.54119982655925	0.050465206435356\\
4.04119982655925	0.035617367706920\\
4.54119982655925	0.020003876720295\\
5.04119982655925	0.009808102345416\\
5.54119982655925	0.004575827594625\\
6.04119982655925	0.001994793417427\\
6.54119982655925	6.478445688091409e-04\\
7.04119982655925	2.202966661771185e-04\\
7.54119982655925	6.549364263079888e-05\\
8.04119982655925	1.694721366623310e-05\\
8.54119982655925	3.663464043785614e-06\\
9.04119982655925	8.208900497877278e-07\\
9.54119982655925	1.289003052949457e-07\\
10.0411998265592	6.784252035288755e-09\\
};
%\addlegendentry{Unbalanced}

\addplot [color=blue, dashed,line width=1pt]
  table[row sep=crcr]{%
3.04119982655925	0.066476061252181\\
3.54119982655925	0.048371777476255\\
4.04119982655925	0.026603993021903\\
4.54119982655925	0.012502422950184\\
5.04119982655925	0.003774179058169\\
5.54119982655925	0.001476196334112\\
6.04119982655925	4.919508981477619e-04\\
6.54119982655925	1.946682266015831e-04\\
7.04119982655925	6.557764172007047e-05\\
7.54119982655925	2.695688135906781e-05\\
8.04119982655925	9.317698622314360e-06\\
8.54119982655925	3.059670895828355e-06\\
9.04119982655925	9.113462536210862e-07\\
9.54119982655925	2.611927238871268e-07\\
10.0411998265592	3.052901967511872e-08\\
};
%\addlegendentry{Balanced}

\end{axis}

\end{tikzpicture}%

\ref{legend}
\caption{Local decoding BER performance of irregular (balanced and unbalanced) codes with the parameters $\kappa=11$, $z=67$, $\gamma_l=3$. Left: low-SNR region; right: high-SNR region.\label{Fig:LocalDec}\vspace{-1.2cm}
}
\end{figure}
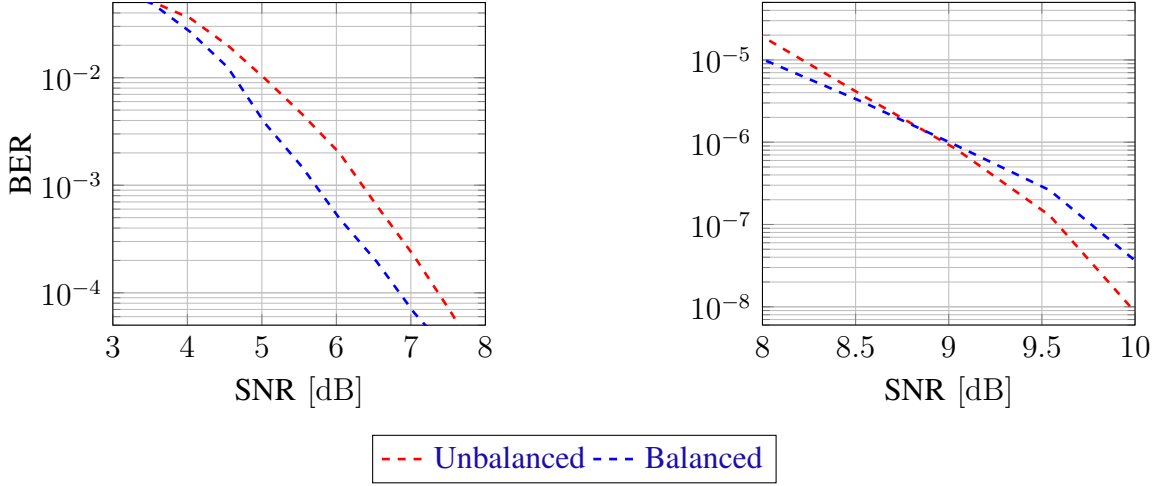

\subsection{Simulations over Partial Response (PR) Channel}

We also perform experiments over the partial response (PR) channel. We use a similar PR setting as in our previous work \cite{EsfTMAG}, which is briefly reviewed here: Our PR setting includes a magnetic recording channel model that incorporates inter-symbol interference in addition to transition jitter noise and electronic noise. The channel density is set to $1.4$, and the equalization target used is $[8,14,2]$. The message is iteratively recovered via a min-sum LDPC decoding algorithm in addition to Bahl Cocke Jelinek Raviv (BCJR) detector based on pattern-dependent noise prediction. The internal iterations inside the LDPC decoder are called local iterations, while a global iteration is the one looping between the detector and the decoder. The decoder performs a specified number of local iterations (fewer if a codeword is reached) between any two successive global iterations. We use $20$ global iterations and $200$ local iterations for our simulations.

Fig.~\ref{Fig:PR} shows the BER comparison between three SC-LDPC code constructions: cutting vector (CV), optimal overlap (OO), and cycle-driven (CD), {over PR channel model with two different level of transition jitter noise}. Since our proposed threshold-driven (TD) construction optimizes the threshold for AWGN channel, we did not incorporate it in our BER evaluation over PR channel. We see that, while all having the same latency and rate, our CD construction enjoys about one order of magnitude performance improvement at SNR$=14$~dB thanks to the dramatically lower number of cycles-6 compared to the CV and OO constructions in the lifted graph. This observation is supported by the fact that short cycles are sub-graphs of the detrimental combinatorial objects over PR channels \cite{EsfTMAG}. {The deeper error floor that is observed in the lower panel is due to lower level of the jitter noise.}

\begin{figure}
\centering
\begin{tikzpicture}

\begin{axis}[%
width=0.4\linewidth,
height=0.36\linewidth,
at={(0,0)},
xmin=12,
xmax=14,
xlabel={SNR [dB]},
ymode=log,
ymin=1e-06,
ymax=1e-01,
yminorticks=true,
xmajorgrids,
ymajorgrids,
yminorgrids,
legend columns=4, 
legend entries={CD,CV,OO},
legend to name= leg,
legend style={
    legend cell align=left, 
    align=left, 
    draw=black,
    /tikz/column 1.5/.style={column sep=0.4pt},
}
]
\addplot [color=red,line width=1pt]
  table[row sep=crcr]{%
12	0.097991858887381\\
12.25   0.057449118046133\\
12.5    0.009176010854817\\
12.75   5.612917232021710e-04\\
13      1.838860244233379e-05\\
13.25   4.464586160108548e-06\\
13.5    6.054274084124830e-06\\
13.75   7.737856173677069e-06\\
14      6.541763907734057e-06\\
};

%\addplot [color=blue,line width=1pt]
%  table[row sep=crcr]{%
%12      0.106743554952510\\
%12.25	0.058937856173677\\
%12.5    0.014042659430122\\
%12.75   0.001427028493894\\
%13      4.754464043419267e-05\\
%13.25   2.113948165508545e-05\\
%13.5    2.802788565011931e-05\\
%13.75   3.072910032398305e-05\\
%14      2.414735413839891e-05\\
%};

\addplot [color=teal,line width=1pt]
  table[row sep=crcr]{%
12	    0.121322116689281\\
12.25   0.095697693351425\\
12.5    0.042374382632293\\
12.75   0.008080615106287\\
13      5.745014925373134e-04\\
13.25   5.146105834464043e-05\\
13.5    3.857476255088196e-05\\
13.75   4.538453188602443e-05\\
14      4.357232021709634e-05\\
};

\addplot [color=magenta,line width=1pt]
  table[row sep=crcr]{%
12	    0.100846132971506\\
12.25   0.049165264586160\\
12.5    0.010606187245590\\
12.75   8.569118046132971e-04\\
13      6.760542740841248e-05\\
13.25   4.291071913161466e-05\\
13.5    5.356037991858887e-05\\
13.75   5.902717815879282e-05\\
14      4.922388059701492e-05\\
};

\end{axis}

\begin{axis}[%
width=0.4\linewidth,
height=0.36\linewidth,
at={(3.4in,0)},
xmin=12.75,
xmax=14,
xlabel={SNR [dB]},
ymode=log,
ymin=1e-07,
ymax=1e-04,
yminorticks=true,
xmajorgrids,
ymajorgrids,
yminorgrids,
legend columns=4, 
legend entries={CD,CV,OO},
legend to name= leg,
legend style={
    legend cell align=left, 
    align=left, 
    draw=black,
    /tikz/column 1.5/.style={column sep=0.4pt},
}
]

\addplot [color=red,line width=1pt]
  table[row sep=crcr]{%
12.75   6.6703e-06\\
13      7.3178e-06\\
13.25   6.0499e-06\\
13.5    3.8768e-06\\
13.75   1.9626e-06\\
14      8.2822e-07\\
};

\addplot [color=teal,line width=1pt]
  table[row sep=crcr]{%
12.75   3.5875e-05\\
13      4.2956e-05\\
13.25   4.1562e-05\\
13.5    2.9508e-05\\
13.75   1.7745e-05\\
14      9.0518e-06\\
};

\addplot [color=magenta,line width=1pt]
  table[row sep=crcr]{%
12.75   5.1881e-05\\
13      5.6793e-05\\
13.25   4.6920e-05\\
13.5    3.0349e-05\\
13.75   1.6607e-05\\
14      7.9593e-06\\
};

\end{axis}

\end{tikzpicture}%

\ref{leg}
\caption{BER performance over PR channel of various code constructions for SC-LDPC codes with parameters $\kappa=11$, $z=67$, $\gamma=3$, $m=1$, and $l=5$. {The transition jitter noise is heavier in the top panel (i.e., $70\%$) compared to the bottom panel (i.e., $50\%$)}}
\label{Fig:PR}
\end{figure}
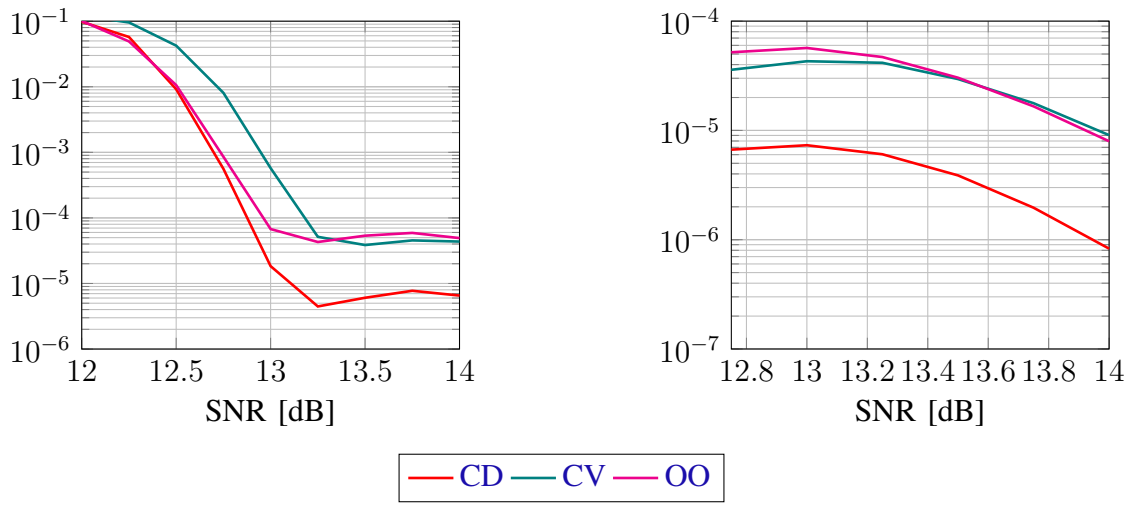

Fig.~\ref{Fig:PR} highlights the importance of optimizing the partitioning matrix to improve reliability of the storage device by an order of magnitude, without any additional cost in the encoding and decoding procedure. One can further improve the performance of the code by changing the code parameters, such as increasing the row weight $\kappa$, column weight $\gamma$, increasing the field size, among others.

\section{Conclusion}
In this paper, we proposed a novel framework to reduce the search space of block LDPC and SC-LDPC codes via only keeping one member from a family of equivalent matrices that share identical finite-length and asymptotic metrics (cycle-6 and thresholds, respectively).
Then, we proposed a design method that identifies all constructions that offer a trade-off between finite-length and asymptotic performances in this reduced search space. Further, we incorporated the block locality feature into our SC-LDPC design, and proposed methods for designing both local CNs and global CNs. Our simulation results verify our theoretical derivations and show the outstanding performance and flexibility of the codes designed using our method.
For future work, one can also incorporate additional constraints over the search space of nonequivalent matrices introduced in this paper, e.g., all columns must be at least of certain weight.\vspace{-0.4cm}

\section{Acknowledgements}
\label{Sec:ACK}

Research supported in part by a grant from ASRC-IDEMA, and in part by a grant from the Israel Science Foundation. The authors would like to thank Lev Tauz and Debarnab Mitra for their assistance in carrying the PR experiments.

%\section{Acknowledgment\vspace{-0.2cm}}
%Research supported in part by a grant from ASRC-IDEMA, NSF CCF-BSF:CIF 1718389, and ISF 2525/19.\vspace{-0.6cm}

%\balance
\bibliographystyle{IEEEtran}
%\vspace{-0.6cm}
\bibliography{IEEEabrv,references}

\newpage

\section*{Appendices}

\begin{algorithm}[H]
    \label{star_bar_sol}
     \SetAlgoLined
     \SetKwProg{Fn}{Function}{ is}{end}
     \KwResult{$\text{list}=\{[n_0,\dots,n_{\alpha-1}]:n_i\geq0\text{ and } \sum_{i=0}^{\alpha-1}n_i=\beta \}$}
     $\text{list}=\{\}$, $a=\alpha$, $b=\beta$\;
     $F(a,b,[])$\;
     \Fn{$F(a,b,\textrm{choice})$}{
        \uIf{$a=1$}{
             $\text{choice} = [\text{choice},b]$\;
           list = list $\cup$ choice\;
         }
         \uElse{  
             \For{$k\in\{0,\dots,b\}$}{
                 $\text{choice}'=[ \text{choice},k]$ \;
                 $F(a-1,b-k,\text{choice}')$\;
             }
       }
    }
\caption{Recursive Algorithm for Solving Star and Bar Problem}
\end{algorithm}

%\subsection{The Algorithm Computing $C_6$}\label{app:c6}

%\begin{algorithm}\label{alg:cycles}
%    \SetAlgoLined
%    \SetKwProg{Fn}{Function}{ is}{end}
%    \Fn{$f(B_0,B_1,C,\kappa,\gamma,z,L)$}{
%        $C_6^1=0$ and $C_6^2=0$\\
%        Construct $R_1$, $R_2$, $c_1$, and $C_2$.\\
%        \For{$\{i_1,i_2,i_3,j_1,j_2,j_3\}\in\{0,\dots,2\gamma-1\}^3\times\{0,\dots,\kappa-1\}^3$}{
%            \If{$R_1[i_1,j_1]=R_1[i_1,j_2]=R_1[i_2,j_2]=R_1[i_2,j_3]=R_1[i_3,j_3]=R_1[i_3,j_1]=1$}{
%                \If{$C_1[i_1,j_1]-C_1[i_1,j_2]+C_1[i_2,j_2]-C_1[i_2,j_3]+C_1[i_3,j_3]-C_1[i_3,j_1]\overset{z}{=}0$}{
%                    $C_6^1=C_6^1+z$\;
%                }
%            }
%        }
     %   \For{$\{i_1,i_2,i_3,j_1,j_2,j_3\}\in\{0,\dots,3\gamma-1\}^3\times\{0,\dots,2\kappa-1\}^3$}{
           % \If{$R_2[i_1,j_1]=R_2[i_1,j_2]=R_2[i_2,j_2]=R_2[i_2,j_3]=R_2[i_3,j_3]=R_2[i_3,j_1]=1$}{
%                \If{$C_2[i_1,j_1]-C_2[i_1,j_2]+C_2[i_2,j_2]-C_2[i_2,j_3]+C_2[i_3,j_3]-C_2[i_3,j_1]\overset{z}{=}0$} {
%                $C_6^2=C_6^2+z$\;
%                
%                }
%            }
%        }
%        Output $C_6=LC_6^1+(L-1)(C_6^2-2C_6^1)$
%    }
%\caption{Efficient Procedure of Enumerating Cycles-$6$ for SC Codes}
%\end{algorithm}
\vspace{1cm}

\begin{lemma}\label{col_perm}
Consider two binary matrices $\mathbf{A}$ and $\mathbf{B}$ with size $\gamma\times\kappa$. Let $\mathcal{O}(\mathbf{A})$ and $\mathcal{O}(\mathbf{B})$ be the set of overlap parameters for these two matrices, respectively. 
Then, $\mathbf{A}$ and $\mathbf{B}$ are column permuted versions of each other if and only if $\mathcal{O}(\mathbf{A})=\mathcal{O}(\mathbf{B})$.
\end{lemma}

\textbf{Proof of Lemma~\ref{col_perm}.} We remind that a degree-$d$ overlap parameter $t_{\{i_1,\dots,i_d\}}$ for a binary matrix is defined as the number of columns such that all the rows with indices in ${\{i_1,\dots,i_d\}}$ have $1$s simultaneously. 
We show in the next lemma that the space of all distinct sets of overlap parameters and the space of all distinct column distributions have a one-to-one correspondence. 

We first prove the first direction, i.e., if $\mathcal{O}(\mathbf{A})=\mathcal{O}(\mathbf{B})$, two matrices $\mathbf{A}$ and $\mathbf{B}$ are column permuted version of each other. To this end, we show that if the order of the columns does not matter, a binary matrix can be uniquely constructed using its overlap parameters. Thus, we conclude that two matrices $\mathbf{A}$ and $\mathbf{B}$ must be column permuted version of each other as they have the same set of overlap parameters (both are column permutation of the uniquely constructed matrix). First, we show how one can successively construct the binary matrix using its overlap parameters. 
For the first row, $t_{\{1\}}$ elements are selected to have value $1$ and the rest are set to $0$. For the second row, $t_{\{2\}}$ elements are selected to have value $1$ such that $t_{\{1,2\}}$ of them belong to the columns that the first row also has $0$ and $(t_{\{2\}}-t_{\{1,2\}})$ of them belong to the columns that the first row has $0$. Then, the rest are set to $0$. The remaining rows are formed successively in a similar fashion without any reference to the specific column order but rather with reference to the relative positions of columns. Therefore, the set of matrices that share the same set of overlap parameters must be column permuted versions of each other.\\
The other direction is more straight forward. In fact, the number of overlaps among a set of rows of a binary matrix does not change via column permutations. Thus, if two binary matrices $\mathbf{A}$ and $\mathbf{B}$ are column permuted version of each other, $\mathcal{O}(\mathbf{A})=\mathcal{O}(\mathbf{B})$.\qed

\vspace{1cm}

\textbf{Proof of Theorem~\ref{theorem:gamma2}.} Equation \eqref{Eq:KExplicit1} is a direct consequence of Lemma~\ref{lemma:gamma2}, and it ensures the distinct column distributions of equivalent matrices are considered only once by imposing the constraint $n_1\leq n_2$. To find the number of nonequivalent binary matrices with $ \gamma=2 $ rows, we first calculate the cardinality of the set $ \mathcal T_{\kappa,2}\subseteq \mathcal S_{\kappa,2} $ defined below
\begin{align}\label{Eq:gamma2form1}
\mathcal T_{\kappa,2} = \left \{ [n_0,n_1,n_2,n_3]\in  \mathcal S_{\kappa,2}\colon n_1= n_2\right \}.
\end{align}
Then, the number of nonequivalent matrices will follow from
\begin{align}\label{Eq:K2}
|\mathcal K_{\kappa,2}| = \left  |\mathcal{T}_{\kappa,2}\right | + \frac12\left| \mathcal S_{\kappa,2}\setminus \mathcal T_{\kappa,2} \right |.
\end{align}
Using \eqref{Eq:gamma2form1}, 
\begin{align*}
	\mathcal T_{\kappa,2}= \bigcup_{i=0}^{\lfloor\kappa/2\rfloor} \left \{ [n_0,n_1,n_2,n_3]\in  \mathcal S_{\kappa,2}\colon n_1= n_2= i\right \}\triangleq \bigcup_{i=0}^{\lfloor\kappa/2\rfloor} \mathcal{Q}_{\kappa,i}.		
	\end{align*}
	By invoking the stars-and-bars method again, $ \left| \mathcal{Q}_{\kappa,i} \right | = \binom{\kappa-2i +2-1}{2-1} = \kappa-2i +1 $. Since the union above is disjoint, we get 
	\begin{align}\label{Eq:SminusT}
	\left|\mathcal T_{\kappa,2}\right | =  \sum_{i=0}^{\lfloor\kappa/2\rfloor} (\kappa-2i +1).
	\end{align}
	We remind that according to \eqref{Eq:StarsnBars}, $\left |\mathcal S_{\kappa,2}\right | =\binom{\kappa+3}{3}$, and combining with
	\eqref{Eq:K2} and \eqref{Eq:SminusT} yields
	\begin{align*}
	|\mathcal K_{\kappa,2}| 
	&= \left |\mathcal T_{\kappa,2}\right | + \frac12\left |\mathcal S_{\kappa,2} \right|- \frac12\left |\mathcal T_{\kappa,2} \right | = \frac12\left( \left|\mathcal S_{\kappa,2}\right |+ \left |\mathcal T_{\kappa,2} \right |\right)= \frac12\left( \sum_{i=0}^{\lfloor\kappa/2\rfloor} (\kappa-2i +1)+\binom{\kappa+3}{3} \right ).
	\end{align*}\qed
	
\vspace{1cm}

\textbf{Proof of Theorem~\ref{theorem:gamma3}. }

The first part is a direct consequence of Lemma~\ref{lemma:gamma3} that ensures the distinct column distributions of equivalent matrices are only considered once by imposing appropriate constraints. For the second part, we find the number of distinct nonequivalent binary matrices with $\gamma=3$ rows. In order to do so, we partition the column distributions in $\mathcal S_{\kappa,3}$ into three classes:
    \begin{itemize}
        \item Class-A column distributions $S_{\kappa,3}^A$: column distributions that are invariant to any row permutation. In other words, for a matrix with column distribution in Class-A, all ($3!$ out of $3!$) row permutations of the matrix results in the same column  distribution.
        \item Class-B column distributions $S_{\kappa,3}^B$: column distributions that are invariant to permutation of one pair of rows. In other words, for a matrix with column distribution in Class-B, $3!/2!=3$ row permutations exist that result in distinct column distributions in Class-B.
        \item Class-C column distributions $S_{\kappa,3}^C$: column distributions that are variant to any row permutation. In other words, for a matrix with column distribution in Class-A, all $3!=6$ row permutations of the matrix result in distinct column distributions in Class-C.
    \end{itemize}
Consequently,
\begin{equation}\label{equ:K_gamma_3_summatio}
    |\mathcal{K}_{\kappa,3}|=|S_{\kappa,3}^A|+|S_{\kappa,3}^B|/3+|S_{\kappa,3}^C|/6.
\end{equation}
We first identify $S_{\kappa,3}^A$ as follows:
\begin{equation*}
    S_{\kappa,3}^A = \left \{\left [n_0,n_1,n_2,n_4,n_6,n_5,n_3,n_7\right ]\in \mathcal S_{\kappa,3}\colon n_1=n_2=n_4, n_6=n_5=n_3  \right \}.
\end{equation*}
Thus,
\begin{equation*}
    S_{\kappa,3}^A = \bigcup_{i,j\in \mathbb N  \colon\atop{ 3(i+j)\leq \kappa}} \left \{\left [n_0,n_1,n_2,n_4,n_6,n_5,n_3,n_7\right ]\in \mathcal S_{\kappa,3} \colon n_1=n_2=n_4=i,n_6=n_5=n_3=j \right \},
\end{equation*}
where the union is disjoint. In view of Lemma~\ref{lemma:StarsnBars},
	\begin{align}\label{Eq:numtype1}
	\begin{split}	
	\left |S_{\kappa,3}^A\right |
	&=\sum_{i,j\in \mathbb N  \colon\atop{ 3(i+j)\leq \kappa}} \left |\left \{[n_0,n_7] \colon n_0+n_7=\kappa-3(i+j) \right \}\right |=\sum_{i,j\in \mathbb N  \colon\atop{ 3(i+j)\leq \kappa}}  (\kappa-3(i+j)+1)=a_\kappa.
	\end{split}
\end{align}
Next, we identify $S_{\kappa,3}^B$. We define $\mathcal T_{\kappa,3}^{1,2}$ as the the set of column distributions that are invariant to swapping the first and second row:
\begin{equation*}
    \mathcal T_{\kappa,3}^{1,2} = \left \{[ n_0,n_1,n_2,n_4,n_6,n_5,n_3,n_7\right ]\in \mathcal S_{\kappa,3}\colon n_1=n_2, n_6=n_5\},
\end{equation*}
and similarly, $\mathcal T_{\kappa,3}^{1,3}$ and $\mathcal T_{\kappa,3}^{2,3}$ can be defined. Note that $S_{\kappa,3}^A$ is a subset of $\mathcal T_{\kappa,3}^{1,2}$, $\mathcal T_{\kappa,3}^{1,3}$, and $\mathcal T_{\kappa,3}^{2,3}$. Therefore,
\begin{equation*}
    S_{\kappa,3}^B=(\mathcal T_{\kappa,3}^{1,2}\setminus S_{\kappa,3}^A)\cup(\mathcal T_{\kappa,3}^{1,3}\setminus S_{\kappa,3}^A)\cup(\mathcal T_{\kappa,3}^{2,3}\setminus S_{\kappa,3}^A).
\end{equation*}
Because of the disjoint property and since $|\mathcal T_{\kappa,3}^{1,2}\setminus S_{\kappa,3}^A|=|\mathcal T_{\kappa,3}^{1,3}\setminus S_{\kappa,3}^A|=|\mathcal T_{\kappa,3}^{2,3}\setminus S_{\kappa,3}^A|$ (due to the symmetry), we have $|S_{\kappa,3}^B|=3|\mathcal T_{\kappa,3}^{1,2}\setminus S_{\kappa,3}^A|=3(|\mathcal T_{\kappa,3}^{1,2}|-a_\kappa)$. Besides
\begin{align*}
	\mathcal T_{\kappa,3}^{1,2}
	=\bigcup_{i,j\in \mathbb N  \colon\atop{ 2(i+j)\leq \kappa}} \{[n_0,n_1,n_2,n_3,n_4,n_5,n_4,n_7]\in \mathcal S_{\kappa,3}:n_1=n_2=i,n_5=n_6=j\},
\end{align*}
where the union is disjoint. In view of Lemma~\ref{lemma:StarsnBars},
\begin{equation*}
	\left |\mathcal T_{\kappa,3}^{1,2}\right |{=}
	\hspace{-0.3cm}\sum_{i,j\in \mathbb N  \colon\atop{ 2(i+j)\leq \kappa}}\hspace{-0.2cm} \left |\left \{[n_0,n_3,n_4,n_7] \colon n_0+n_3+n_4+n_7=\kappa-2(i+j) \right \}\right |=\hspace{-0.3cm}\sum_{i,j\in \mathbb N  \colon\atop{ 2(i+j)\leq \kappa}}\hspace{-0.2cm}  \binom{\kappa-2(i+j)+3}{3}.
\end{equation*}
Thus, the number column distributions in Class-B is:  
\begin{align}\label{Eq:numtype2}
	|\mathcal{S}_{\kappa,3}^B|=3\left(\sum_{i,j\in \mathbb N  \colon\atop{ 2(i+j)\leq \kappa}}  \binom{\kappa-2(i+j)+3}{3}-A_\kappa\right)=3b_\kappa.
\end{align}
Finally, we identify $\mathcal{S}_{\kappa,3}^C$, i.e., column distributions that are variant to any permutations. We remind that the total number of column distributions is $ \left |\mathcal{S}_{\kappa,3}\right | $, $a_k$ of them belong to Class-A, and $3b_k$ of them belong to Class-B. As a result,
\begin{equation}\label{Eq:numtype3}
    |\mathcal{S}_{\kappa,3}^C|=\binom{\kappa+7}{7}-a_k-3b_k=6c_k.
\end{equation}
Combining \eqref{equ:K_gamma_3_summatio}-\eqref{Eq:numtype3} completes the proof.\qed

\vspace{1cm}

\textbf{Proof of Proposition~\ref{prop:gamma2searcharea}. }If matrix $\mathbf{P}_C$ has an all-zero row, then its column distribution $ (n_0,n_1,n_2,n_3) $ satisfies $ n_3=0 $. Further, if matrix $P_C$ has an all-zero row, $n_1=0$. This is because to have an all-zero row in $P_C$ either $n_1$ or $n_2$ must be zero, and by definition $ n_1\leq  n_2 $. Thus, using the stars-and-bars method, there exists $\binom{\kappa +2-1}{2-1}= \kappa+1$ column distributions for $\mathbf{P}_C$ among the nonequivalent choices that have at least one all-zero row. Similarly, there exists $\kappa+1 $ column distributions for $\mathbf{P}_C$ among the nonequivalent choices that have at least one all-one row. These two subsets of column distributions have only one common member which is a $2\times\kappa$ matrix with one all-zero row and one all-one row. Finally, using the principle of inclusion and exclusion, the number of nonequivalent matrices that do not have all-zero or all-one rows is $|\mathcal{K}_{\kappa,2}|-2(\kappa+1)+1$.\qed

\vspace{1cm}

\textbf{Proof of Proposition~\ref{prop:gamma3searcharea}. }First, we identify the column distributions that yield matrices with all-zero rows. Since we already consider nonequivalent matrices, a single all-zero row is sufficient for the matrix to be invalid for our purpose. Without loss of generality, we can assume the first row is zero, and thus $ n_4=n_5=n_6=n_7=0$, and the other possibilities can be considered as row permutation of this setting. Next, we quantify the possible realizations for $[n_0,n_1,n_2,n_3,0,0,0,0]$ that result in nonequivalent matrices. This corresponds to discarding the first row that is all-zero, and identifying the cardinality of the set of column distributions for all distinct $2\times\kappa$ nonequivalent binary matrices which is $|\mathcal{K}_{\kappa,2}|$ and is given in Theorem~\ref{theorem:gamma2}. Similarly, the number of column distributions that result in nonequivalent matrices with all-one rows is $|\mathcal{K}_{\kappa,2}|$. Lastly, there are $\kappa+1$ column distribution that correspond to matrices with both a full-zero row (the first row without loss of generality) and a full-one row (the second row without loss of generality), which are realizations in form of $[0,0,n_2,n_3,0,0,0,0,0]$ and have $\kappa+1$ possibilities.\qed

\vspace{1cm}

\textbf{Proof of Proposition~\ref{Prop:thU<thB}.}

\begin{fact}\label{Fact: Threshold monotonicity}
Let $\sigma^*_1$ and $\sigma^*_2$ be the EXIT thresholds of  $(\gamma_1,\kappa_1)$-regular and $(\gamma_2,\kappa_2)$-regular protographs.  If $\kappa_1=\kappa_2$ and $\gamma_1\leq\gamma_2$, then $\sigma^*_1\leq\sigma^*_2$, and if $\gamma_1=\gamma_2$ and $\kappa_1\leq\kappa_2$, then $\sigma^*_1\geq\sigma^*_2$.
\end{fact}

Let $\nu=a\gamma_l +b$. Consider a $ (\gamma_l-1,\kappa-a) $-regular protograph. Assume that we apply \eqref{Eq:J VN} and \eqref{Eq:J CN} on this protograph, and let $ x_\ell(\sigma) $ and $ u_\ell(\sigma) $ denote the resulting VN$\rightarrow$CN and CN$\rightarrow$VN EXIT values at iteration $ \ell $, respectively, given the channel parameter $\sigma$. We construct $\overline{\mathbf{B}}_{B}$ as follows:
\begin{align*} 
\overline{\mathbf{B}}_B=&\left(\begin{array}{l:l:}
\mathbf{1}_{\gamma_l\times(\kappa-\nu)}-\mathbf{Q}(\gamma_l,\kappa-\nu;1)& 
\mathbf{1}_{\gamma_l\times b}-\mathbf{S}(\gamma_l,b)
\end{array}\right. \\
&\hspace{2mm}\left.\begin{array}{l:l:l:l:l}
\mathbf{Q}(\gamma_l,\kappa-\nu,1)&
\mathbf{S}(\gamma_l,b)&\mathbf{Q}(\gamma_l,a,\gamma_l) &
\dots &
\mathbf{Q}(\gamma_l,a,1)
\end{array} \right).
\end{align*}
In other words, $ \overline{\mathbf{B}}_B$ is obtained from $\mathbf{B}_B$ by: 1) replacing the leftmost $\kappa-\nu$ entries in the first row with zeros such that all VNs are $\gamma_l-1$ regular, and 2) adding $\kappa-\nu+b$ columns of degree $1$ such that all CNs are $\kappa-a$ regular. We call the added degree-$1$ columns (VNs) ``auxiliary VNs" (see Fig.~\ref{Fig:Th_Proposition} for an example with $\kappa=5,\gamma_l=3,\nu=4$). Thus,  $\overline{\mathbf{B}}_B$ is $(\gamma_l-1,\kappa-a)$-regular except for the auxiliary VNs.
Next, we apply \eqref{Eq:J VN} and \eqref{Eq:J CN} on $\overline{\mathbf{B}}_B$ with a channel parameter $\sigma$ for non-auxiliary VNs, while the auxiliary VNs pass through a channel with a parameter $\sigma_\ell $ that changes in every iteration $\ell$ in a way that $ J(\sigma_\ell)=x_\ell(\sigma)$. It follows that the EXIT values passing over all edges of $\overline{\mathbf{B}}_B$ equal to those passing over a $ (\gamma_l-1,\kappa-a) $-regular protograph, i.e., $ x_\ell(\sigma) $ and $ u_\ell(\sigma) $ for VN$\rightarrow$CN and CN$\rightarrow$VN messages, respectively. We match the edges in ${\mathbf{B}}_B$ to the edges in $\overline{\mathbf{B}}_B$ as follows. The edges connecting the $\nu$ rightmost columns in ${\mathbf{B}}_B$ match their identical edges in $\overline{\mathbf{B}}_B$, and the edges connecting bottom-most $\gamma_l-1$ CNs with the leftmost $\kappa-\nu$ VNs in ${\mathbf{B}}_B$ match their identical edges in $\overline{\mathbf{B}}_B$ as well. Finally, the edges connecting the top CN with the leftmost $\kappa-\nu$ VNs each matches one arbitrary edge connected to an auxiliary VN (see Fig.~\ref{Fig:Th_Proposition}).
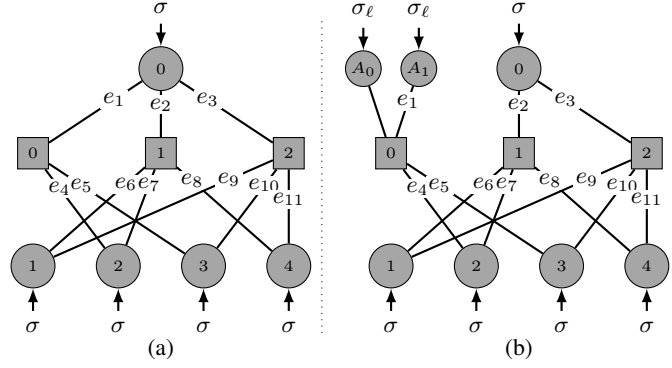
\begin{figure}
    \centering
    \begin{tikzpicture}[>=latex]
\tikzstyle{cnode}=[rectangle,draw,fill=gray!70!white,minimum size=4mm]
\tikzstyle{vnode}=[circle,draw,fill=gray!70!white,minimum size=3mm]
\pgfmathsetmacro{\x}{1.7}
\pgfmathsetmacro{\y}{1.5}

	\foreach \c in {0,1,2}
	{
		\node[cnode] (c\c) at (\c*\x,0) {\tiny$\c$};	
	}
    \foreach \v in {1,2,3,4}
	{
		\pgfmathtruncatemacro{\k}{\v -1}
		\node[vnode] (v\v) at (\k*\x*2/3,-\y) {\tiny$\v$};	
		\node (ch\v) [below=2*\y mm of v\v] {\footnotesize $\sigma$};
		\draw [thick,->] (ch\v)--(v\v);
	}
	\node[vnode] (v0) at (\x,0.75*\y) {\tiny $0$};	
    \node (ch0) [above=2*\y mm of v0] {\footnotesize $\sigma$};
    \draw [thick,->] (ch0)--(v0);
    %graph edges
    \draw [thick] (c0)--(v2) node[fill=white,inner sep=1pt,pos = 0.24] {\footnotesize $e_4$};
    \draw [thick] (c0)--(v3) node[fill=white,inner sep=1pt,pos = 0.24] {\footnotesize $e_5$};
    \draw [thick] (c0)--(v0) node[fill=white,inner sep=1pt,pos = 0.7] {\footnotesize $e_1$};
    \draw [thick] (c1)--(v0) node[fill=white,inner sep=1pt,pos = 0.7] {\footnotesize $e_2$};
    \draw [thick] (c1)--(v1) node[fill=white,inner sep=1pt,pos = 0.2] {\footnotesize $e_6$};
    \draw [thick] (c1)--(v2) node[fill=white,inner sep=1pt,pos = 0.2] {\footnotesize $e_7$};
    \draw [thick] (c1)--(v4) node[fill=white,inner sep=1pt,pos = 0.16] {\footnotesize $e_8$};
    \draw [thick] (c2)--(v0) node[fill=white,inner sep=1pt,pos = 0.7] {\footnotesize $e_3$};
    \draw [thick] (c2)--(v1) node[fill=white,inner sep=1pt,pos = 0.2] {\footnotesize $e_9$};
    \draw [thick] (c2)--(v3) node[fill=white,inner sep=1pt,pos = 0.2] {\footnotesize $e_{10}$};
    \draw [thick] (c2)--(v4) node[fill=white,inner sep=1pt,pos = 0.4] {\footnotesize $e_{11}$};
    \node [below=14*\y mm of c1] {\footnotesize(a)};
    \begin{scope}[shift={(2.8*\x,0)}]
    \foreach \c in {0,1,2}
	{
		\node[cnode] (c\c) at (\c*\x,0) {\tiny $\c$};	
	}
    \foreach \v in {1,2,3,4}
	{
		\pgfmathtruncatemacro{\k}{\v -1}
		\node[vnode] (v\v) at (\k*\x*2/3,-\y) {\tiny $\v$};	
		\node (ch\v) [below=2*\y mm of v\v] {\footnotesize $\sigma$};
		\draw [thick,->] (ch\v)--(v\v);
	}
	\node[vnode] (v0) at (\x,0.75*\y) {\tiny $0$};	
    \node (ch0) [above=2*\y mm of v0] {\footnotesize $\sigma$};
    \draw [thick,->] (ch0)--(v0);

    \node (z) [inner sep = 0pt]at (c0|-v0) {};
    \node[vnode,left = 1mm of z,inner sep = 1pt] (a0) {\tiny $A_0$};	
    \node (cha0) [above=2*\y mm of a0] {\footnotesize $\sigma_\ell$};
    \draw [thick,->] (cha0)--(a0);
    \node[vnode,right = 1mm of z,inner sep = 1pt] (a1) {\tiny $A_1$};	
    \node (cha1) [above=2*\y mm of a1] {\footnotesize $\sigma_\ell$};
    \draw [thick,->] (cha1)--(a1);
    %graph edges
    \draw [thick] (c0)--(a1) node[fill=white,inner sep=1pt,pos = 0.7] {\footnotesize $e_1$};
    \draw [thick] (c0)--(a0);
    \draw [thick] (c0)--(v2) node[fill=white,inner sep=1pt,pos = 0.24] {\footnotesize $e_4$};
    \draw [thick] (c0)--(v3) node[fill=white,inner sep=1pt,pos = 0.24] {\footnotesize $e_5$};
    \draw [thick] (c1)--(v0) node[fill=white,inner sep=1pt,pos = 0.7] {\footnotesize $e_2$};
    \draw [thick] (c1)--(v1) node[fill=white,inner sep=1pt,pos = 0.2] {\footnotesize $e_6$};
    \draw [thick] (c1)--(v2) node[fill=white,inner sep=1pt,pos = 0.2] {\footnotesize $e_7$};
    \draw [thick] (c1)--(v4) node[fill=white,inner sep=1pt,pos = 0.16] {\footnotesize $e_8$};
    \draw [thick] (c2)--(v0) node[fill=white,inner sep=1pt,pos = 0.7] {\footnotesize $e_3$};
    \draw [thick] (c2)--(v1) node[fill=white,inner sep=1pt,pos = 0.2] {\footnotesize $e_9$};
    \draw [thick] (c2)--(v3) node[fill=white,inner sep=1pt,pos = 0.2] {\footnotesize $e_{10}$};
    \draw [thick] (c2)--(v4) node[fill=white,inner sep=1pt,pos = 0.4] {\footnotesize $e_{11}$};
    \node (b) [below=14*\y mm of c1] {\footnotesize(b)};
    \end{scope}
    \node (Y1) [left=1mm of cha0] {};
    \node (Y0) at(Y1|-b) {};
    \draw [dotted] (Y1)--(Y0);

    \end{tikzpicture}
    \caption{Graph constructions for proof of Proposition~\ref{Prop:thU<thB}, with $\kappa=5,\gamma_l=3,\nu=4$: (a) corresponds to ${\mathbf{B}}_B$ and (b) corresponds to $\overline{\mathbf{B}}_B$. $A_0$ and $A_1$ are auxiliary VNs. The edge matching is illustrated via edge labels $\{e_i\}_{i=1}^{11}$.}
    \label{Fig:Th_Proposition}
\end{figure}

Given a channel parameter $ \sigma $, let $ y_\ell(\sigma,e) $ and $ w_\ell(\sigma,e) $ be the VN$\rightarrow$CN and CN$\rightarrow$VN EXIT values, respectively, over some edge $ e $ in the protograph ${\mathbf{B}}_B$. 
From the monotonicity of \eqref{Eq:J VN} and \eqref{Eq:J CN} in their arguments and in node degrees, it can be shown by mathematical induction that for any $\sigma$ and every edge $ e $ 
\begin{align}
\label{Eq:y>x}
y_\ell(\sigma,e) \geq x_\ell(\sigma),\quad w_\ell(\sigma,e) \geq u_\ell(\sigma),\quad \forall \ell\geq 0\,.
\end{align}

If we mark $ \sigma^*(d_v,d_c) $ as the asymptotic threshold of a regular $ (d_v,d_c) $ protograph, then \eqref{Eq:y>x} implies that if the channel parameter satisfies $ \sigma<\sigma^*(\gamma_l-1,\kappa-a) $ then the EXIT algorithm over ${\mathbf{B}}_B$ will converge to $ 1 $, thus
\begin{align}\label{Eq:H_B>}
 \sigma^*({\mathbf{B}}_B) \geq \sigma^*(\gamma_l-1,\kappa-a)\,.
\end{align}
From the sub-matrix lemma in \cite[Lemma~1]{RamCass18b} we have \begin{align}\label{Eq:H_U<}
\sigma^*({\mathbf{B}}_U)\leq \sigma^*(\gamma_l-1,\nu)\,.
\end{align}
Since $ \kappa-a\leq \nu $, combining \eqref{Eq:H_B>}--\eqref{Eq:H_U<} with Fact~\ref{Fact: Threshold monotonicity}, which holds since $\kappa-a = \kappa-\lfloor\nu/\gamma_l\rfloor\leq \nu$, completes the proof.\qed

\vspace{1cm}
\textbf{Proof of Proposition~\ref{Prop:C6U<C6B}.} Consider a binary matrix $\mathbf{B}$ with $\gamma_l=3$ rows. The number of cycles-$6$ can be expressed in terms of the overlap parameters of proto-matrix $\mathbf{B}$ as described in (\ref{eq:cyclesA}) and (\ref{equ_A}) as follows:
\begin{equation*}
    F(\mathbf{B})=A(t_{\{1,2,3\}},t_{\{1,2\}},t_{\{1,3\}},t_{\{2,3\}})\,.
\end{equation*}
According to our constructions, no two zeros (out of the $\nu$ zeros) are located in the same column and thus
$ t_{\{1,2,3\}} = \kappa-\nu=\kappa-3a-b\geq 1$.
In the balanced construction, we have $t_{\{1,2\}}=\kappa-2a-b$, $t_{\{1,3\}}=\kappa-2a-(b>0)$, and $t_{\{2,3\}}=\kappa-2a-(b>1)$, where $(\mathrm{cond})$ is $1$ if $\mathrm{cond}$ is true and $0$ otherwise. Thus,
\begin{equation}\label{Eq:F_Balgamma3}
\begin{split}
    F(\mathbf{B}_B)=&(\kappa-\nu)(\kappa-\nu-1)(\kappa-2a-(b>1)-2)\\
    +&(\kappa-\nu)(a+(b>1))(\kappa-2a-(b>1)-1)\\
    +&a(\kappa-\nu)(\kappa-2a-(b>1)-1)a(\kappa-\nu)(\kappa-2a-(b>1)).
\end{split}
\end{equation}
In the unbalanced construction, we have $t_{\{1,2\}}=t_{\{1,3\}}=\kappa-\nu$ and  $t_{\{2,3\}}=\kappa$. Thus,
\begin{align}\label{Eq:F_Unbalgamma3}
    F(\mathbf{B}_U) =(\kappa-\nu)(\kappa-\nu-1)(\kappa-2)\,.
\end{align}
Comparing \eqref{Eq:F_Balgamma3} and \eqref{Eq:F_Unbalgamma3} completes the proof.\qed

\vspace{1cm}
\textbf{Proof of Proposition~\ref{pro_FB_FU_gamma4}.} Consider a binary matrix $\mathbf{B}$ with $\gamma_l=4$ rows. The number of cycles-$6$ in $\mathbf{B}$ is
\begin{equation*}
\begin{split}
F(\mathbf{B})
=&A(t_{\{1,2,3\}},t_{\{1,2\}},t_{\{1,3\}},t_{\{2,3\}})+A(t_{\{1,2,4\}},t_{\{1,2\}},t_{\{1,4\}},t_{\{2,4\}})\\
+&A(t_{\{1,3,4\}},t_{\{1,3\}},t_{\{1,4\}},t_{\{3,4\}})+A(t_{\{2,3,4\}},t_{\{2,3\}},t_{\{2,4\}},t_{\{3,4\}}).
\end{split}
\end{equation*}

Again, zeros are never located in the same column according to our constructions. In the balanced construction, we have $t_{\{1,2\}}=t_{\{1,3\}}=t_{\{1,4\}}=t_{\{2,3\}}=t_{\{2,4\}}=t_{\{3,4\}}=\kappa-\nu/2 $, and $t_{\{1,2,3\}}=t_{\{1,2,4\}}=t_{\{1,3,4\}}=t_{\{2,3,4\}}=\kappa-3\nu/4$, thus
\begin{align}
\label{Eq:F_Balgamma4}
\begin{split}
    F(\mathbf{B}_B)
    =&4(\kappa-3\nu/4)(\kappa-3\nu/4-1)(\kappa-\nu/2-2)+4(\kappa-3\nu/4)(\nu/4)(\kappa-\nu/2-1)\\
    +&4(\nu/4)(\kappa-3\nu/4)(\kappa-\nu/2-1)+4(\nu/4)(\nu/4)(\kappa-\nu/2).
\end{split}
\end{align}
In the unbalanced construction, we have $t_{\{1,2,3\}}=t_{\{1,2,4\}}=t_{\{1,3,4\}}=\kappa-\nu$, $t_{\{2,3,4\}}=\kappa$, $t_{\{1,2\}}=t_{\{1,3\}}=t_{\{1,4\}}=\kappa-\nu$, and  $t_{\{2,3\}}=t_{\{2,4\}}=t_{\{3,4\}}=\kappa$, thus
\begin{align}
\label{Eq:F_Unbalgamma4}
\begin{split}
    F(\mathbf{B}_U)
    &\!=\!3(\kappa\!-\!\nu)(\kappa\!-\!\nu\!-\!1)(\kappa\!-\!2)\!+\!\kappa(\kappa\!-\!1)(\kappa\!-\!2).
\end{split}
\end{align}
In view of (\ref{Eq:F_Balgamma4}) and (\ref{Eq:F_Unbalgamma4}),  $F(\mathbf{B}_B)-F(\mathbf{B}_U)=\nu^2(3/2-\nu)<0
$ since $\nu \geq \gamma_l =4$.\qed

\end{document}